\documentclass{amsart}%
\usepackage{amsmath}
\usepackage{amssymb}
\usepackage{amsfonts}
\usepackage{graphicx}%
\newtheorem{theorem}{Theorem}
\theoremstyle{plain}
\newtheorem{acknowledgement}{Acknowledgement}

\newtheorem{corollary}{Corollary}

\newtheorem{definition}{Definition}
\newtheorem{example}{Example}

\newtheorem{lemma}{Lemma}
\newtheorem{notation}{Notation}

\newtheorem{proposition}{Proposition}
\newtheorem{remark}{Remark}

\numberwithin{equation}{section}
\begin{document}
\title[$p$-adic String Amplitudes and Local Zeta Functions]{Regularization of $p$-adic String Amplitudes, and Multivariate Local Zeta Functions}
\author{Miriam Bocardo-Gaspar}
\address{Centro de Investigaci\'{o}n y de Estudios Avanzados del Instituto
Polit\'{e}cnico Nacional\\
Departamento de Matem\'{a}ticas, Unidad Quer\'{e}taro\\
Libramiento Norponiente \#2000, Fracc. Real de Juriquilla. Santiago de
Quer\'{e}taro, Qro. 76230\\
M\'{e}xico.}
\email{mbocardo@math.cinvestav.mx}
\author{H. Garc\'{\i}a-Compe\'{a}n}
\address{Centro de Investigacion y de Estudios Avanzados del I.P.N., Departamento de
F\'{\i}sica, Av. Instituto Politecnico Nacional 2508, Col. San Pedro
Zacatenco, Mexico D.F., C.P. 07360, Mexico}
\email{compean@fis.cinvestav.mx}
\author{W. A. Z\'{u}\~{n}iga-Galindo}
\address{Centro de Investigaci\'{o}n y de Estudios Avanzados del Instituto
Polit\'{e}cnico Nacional\\
Departamento de Matem\'{a}ticas, Unidad Quer\'{e}taro\\
Libramiento Norponiente \#2000, Fracc. Real de Juriquilla. Santiago de
Quer\'{e}taro, Qro. 76230\\
M\'{e}xico.}
\email{wazuniga@math.cinvestav.edu.mx}
\thanks{The third author was partially supported by Conacyt Grant No. 250845.}

\begin{abstract}
We prove that the $p$-adic Koba-Nielsen type string amplitudes are bona fide
integrals. We attach to these amplitudes Igusa-type integrals depending on
several complex parameters and show that these integrals admit meromorphic
continuations as rational functions. Then we use these functions to regularize
the Koba-Nielsen amplitudes. As far as we know, there is no a similar result
for the Archimedean Koba-Nielsen amplitudes. We also discuss the existence of
divergencies and the connections with multivariate Igusa's local zeta functions.

\end{abstract}
\maketitle

\section{Introduction}

This article aims to discuss some connections between $p$-adic string
amplitudes and $p$-adic local zeta functions (also called Igusa's local zeta
functions). In the 80s, Volovich posed the conjecture that the space-time has
a non-Archimedean structure at the level of the Planck scale and initiated the
$p$-adic string theory \cite{Vol}, see also \cite[Chapter 6]{Var},
\cite{V-V-Z}. Volovich noted that the integral expression for the Veneziano
amplitude of the open bosonic string can be generalized to a $p$-adic integral
and to an adelic integral giving rise to non-Archimedean Veneziano amplitudes.
Then Freund and Witten established (formally) that the ordinary Veneziano and
Virasoro-Shapiro four-particle scattering amplitudes can be factored in terms
of an infinite product of non-Archimedean string amplitudes \cite{F-W}, see
also \cite{Arefeva}. In $p$-adic string theory, justly as in the case of usual
string theory, the problem of computing the scattering amplitudes in the
perturbative theory is formulated as follows. Consider a certain number $N$ of
vertex operators $\mathcal{V}_{1}(x_{1})$, $\mathcal{V}_{2}(x_{2})$,$\dots$,
$\mathcal{V}_{N}(x_{N})$, which are inserted at arbitrary ordered points
$x_{1}$, $x_{2}$,$\dots$, $x_{N}$ in the world-sheet manifold $\Sigma_{g,N}$.
This manifold is a Riemann surface of genus $g$ and $N$ ordered marked points.
The computation of these amplitudes implies to perform formal integrations
over the moduli space of Riemann surfaces $\mathcal{M}_{g,N}$. Tree amplitudes
correspond to the moduli space of $N$ ordered points on $\Sigma_{0,N}$. In
this article we will workout the tree level amplitudes.

As a consequence of the ensuing interest on $p$-adic models of quantum field
theory, which is motivated by the fact that these models are exactly solvable,
there is a large list of $p$-adic type Feynman and string amplitudes that are
related with local zeta functions of Igusa-type, and it is interesting to
mention that it seems that the mathematical community working on local zeta
functions is not aware of this fact, see e.g. \cite{Arefeva}-\cite{Bleher},
\cite{Bogner}, \cite{B-F-O-W}-\cite{B-F}, \cite{Fra-Okada}-\cite{F-W},
\cite{Goshal}, \cite{Hloseuk-Spect}, \cite{Lerner}-\cite{Lerner2},
\cite{Marcolli}, \cite{Smirnov1}-\cite{Smirnov2}, \cite{Speer1} and the
references therein.

The connections between Feynman amplitudes and local zeta functions are very
old and deep. Let us mention that the works of Speer \cite{Speer} and Bollini,
Giambiagi and Gonz\'{a}lez Dom\'{\i}nguez \cite{B-G-Gonzalez-Dom} on
regularization of Feynman amplitudes in quantum field theory are based on the
analytic continuation of distributions attached to complex powers of
polynomial functions in the sense of Gel'fand and Shilov \cite{G-S}, see also
\cite{Belkale}-\cite{Bleher}, \cite{Bogner}, \cite{Marcolli}, among others.
There are several types of local zeta functions, for instance $p$-adic,
Archimedean, topological and motivic, among others, see e.g. \cite{Denef}%
-\cite{DL1}, \cite{G-S}, \cite{Igusa0}-\cite{Igusa}, and the references
therein. In the Archimedean setting, the local zeta functions were introduced
in the 50s by Gel'fand and Shilov. The main motivation was that the
meromorphic continuation of Archimedean local zeta functions implies the
existence of fundamental solutions (i.e. Green functions) for differential
operators with constant coefficients. This fact was established,
independently, by Atiyah \cite{Atiyah} and Bernstein \cite{Ber}. It is
important to mention here, that in the $p$-adic framework, the existence of
fundamental solutions for pseudodifferential operators is also a consequence
of the fact that the Igusa local zeta functions admit a meromorphic
continuations, see \cite[Chapter 5]{Zuniga-LNM-2016}, \cite[Chapter
10]{KKZuniga}. This analogy turns out to be very important in the rigorous
construction of quantum scalar fields in the $p$-adic setting, see
\cite{M-V-Zuniga} and the references therein.

In the 60s, Weil studied local zeta functions, in the Archimedean and
non-Archimedean settings, in connection with the Poisson-Siegel formula
\cite{We}. In the 70s, Igusa developed a uniform theory for local zeta
functions in characteristic zero \cite{Igusa0}-\cite{Igusa}. In the $p$-adic
setting, the local zeta functions are connected with the number of solutions
of polynomial congruences mod $p^{m}$ and with exponential sums mod $p^{m}$.
Recently Denef and Loeser introduced the motivic zeta functions which
constitute a vast generalization of $p$-adic local zeta functions
\cite{DL1}-\cite{DL2}.

Take $N\geq4$ and $s_{ij}\in\mathbb{C}$ satisfying $s_{ij}=s_{ji}$ for $1\leq
i<j\leq N-1$. In this article we study the following multivariate Igusa-type
zeta function:%
\begin{equation}
Z^{(N)}\left(  \boldsymbol{s}\right)  =%
{\displaystyle\int\limits_{\mathbb{Q}_{p}^{N-3}\smallsetminus\Lambda}}
{\displaystyle\prod\limits_{i=2}^{N-2}}
\left\vert x_{i}\right\vert _{p}^{s_{1i}}\left\vert 1-x_{i}\right\vert
_{p}^{s_{(N-1)i}}\text{ }%
{\displaystyle\prod\limits_{2\leq i<j\leq N-2}}
\left\vert x_{i}-x_{j}\right\vert _{p}^{s_{ij}}%
{\displaystyle\prod\limits_{i=2}^{N-2}}
dx_{i}, \label{zeta_funtion_string}%
\end{equation}
where $\boldsymbol{s}=\left(  s_{ij}\right)  \in\mathbb{C}^{D}$, here $D$
denotes the total number of possible subsets $\left\{  i,j\right\}  $, $%
{\textstyle\prod\nolimits_{i=2}^{N-2}}
dx_{i}$ is the normalized Haar measure of $\mathbb{Q}_{p}^{N-3}$, and%
\[
\Lambda:=\left\{  \left(  x_{2},\ldots,x_{N-2}\right)  \in\mathbb{Q}_{p}%
^{N-3};\text{ }%
{\displaystyle\prod\limits_{i=2}^{N-2}}
x_{i}\left(  1-x_{i}\right)  \text{ }%
{\displaystyle\prod\limits_{2\leq i<j\leq N-2}}
\left(  x_{i}-x_{j}\right)  =0\right\}  .
\]
We study integrals of type (\ref{zeta_funtion_string}) by using the theory of
local zeta functions, in this framework, it is not convenient, neither
necessary, to assume some algebraic dependency between the variables $s_{ij}$.
We call this type of integrals $p$\textit{-adic open string }$N$\textit{-point
zeta functions} because they appeared in connection with the $p$-adic open
string $N$-tachyon tree amplitudes, see e.g. \cite{B-F-O-W}-\cite{B-F},
\cite{F-O}-\cite{F-W}, \cite{Ghoshal:2000dd}, \cite{Hloseuk-Spect},
\cite{Lerner}, \cite{Smirnov2}, and the references therein. In all the
published literature about $p$-adic string amplitudes, integrals of type
(\ref{zeta_funtion_string}) have been used without considering the convergence
properties of them, i.e. the problem of the regularization of $p$-adic open
string $N$-tachyon amplitudes has not been considered before. In the light of
the theory of local zeta functions,\ the possible convergence of integrals of
type (\ref{zeta_funtion_string}) is a new and remarkable aspect. Theorem
\ref{TheoremA}, the main result of this article, establishes that the $p$-adic
open string $N$-point zeta function is a holomorphic function in a certain
domain of $\mathbb{C}^{D}$ and that it admits an analytic continuation to
$\mathbb{C}^{D}$ (denoted as $Z^{\left(  N\right)  }\left(  \boldsymbol{s}%
\right)  $) as a rational function in the variables $p^{-s_{ij}}%
,i,j\in\left\{  1,\ldots,N-1\right\}  $. In addition, if $\boldsymbol{s}%
=\left(  s_{ij}\right)  \in\mathbb{R}^{D}$, with $s_{ij}\geq0$ for
$i,j\in\left\{  1,\ldots,N-1\right\}  $, then the integral in
(\ref{zeta_funtion_string}) diverges to $+\infty$.

The typical approach to establish the existence of a meromorphic continuation
for an integral of Igusa-type (which is holomorphic in a certain domain) is
via Hironaka's resolution of singularities theorem, see e.g. \cite[Chapters 3,
5, 8]{Igusa}. Roughly speaking Hironaka's resolution theorem provides a finite
sequence of changes of variables (blow-ups) that allows to express an
Igusa-type integral as a linear combination of integrals involving monomials,
for this type of integrals the existence of an analytic continuation is easy
to show. If the analyticity of the initial Igusa-type integral is unknown,
then the above described approach cannot be used. In addition, nowadays,
Hironaka's resolution theorem is only valid in characteristic zero. Here, we
use an approach inspired in the calculations presented in \cite{B-F-O-W} and
in the Igusa $p$-adic stationary phase formula, see \cite[Theorem
10.2.1]{Igusa}, \cite{Zuniga1}-\cite{Zuniga2}. This approach works on
non-Archimedean local fields of arbitrary characteristic, for instance in
$\mathbb{F}_{q}\left(  \left(  t\right)  \right)  $, the field of formal
Laurent series over a finite field $\mathbb{F}_{q}$. In addition, our approach
provides an algorithm for computing the $p$-adic open string $N$-point zeta
functions\textit{.}

Take $\phi\left(  x_{2},\ldots,x_{N-2}\right)  $ a locally constant function
with compact support, then
\begin{gather*}
Z_{\phi}^{(N)}(\boldsymbol{s})=\\%
{\displaystyle\int\limits_{\mathbb{Q}_{p}^{N-3}\smallsetminus\Lambda}}
\phi\left(  x_{2},\ldots,x_{N-2}\right)
{\displaystyle\prod\limits_{i=2}^{N-2}}
\left\vert x_{i}\right\vert _{p}^{s_{1i}}\left\vert 1-x_{i}\right\vert
_{p}^{s_{\left(  N-1\right)  i}}%
{\displaystyle\prod\limits_{2\leq i<j\leq N-2}}
\left\vert x_{i}-x_{j}\right\vert _{p}^{s_{i}{}_{j}}%
{\displaystyle\prod\limits_{i=2}^{N-2}}
dx_{i},
\end{gather*}
for $\operatorname{Re}(s_{ij})>0$ for any $ij$, is a multivariate Igusa local
zeta function. In characte\-ristic zero, a general theory for this type of
local zeta functions was elaborated by Loeser in \cite{Loeser}. In particular,
these local zeta functions admit analytic continuations as rational functions
of the variables $p^{-s_{ij}}$. If we take $\phi$ to be the characteristic
function of\ $B_{r}^{N-3}$, the ball centered at the origin with radius
$p^{r}$, the dominated convergence theorem and Theorem \ref{TheoremA}, imply
that $\lim_{r\rightarrow\infty}Z_{B_{r}^{N-3}}^{(N)}(\boldsymbol{s}%
)=Z^{(N)}\left(  \boldsymbol{s}\right)  $ for any $\boldsymbol{s} $ in the
natural domain of $Z^{(N)}\left(  \boldsymbol{s}\right)  $.

In \cite{B-F-O-W}, Brekke, Freund, Olson and Witten work out the $N$-point
amplitudes in explicit form and investigate how these can be obtained from an
effective Lagrangian. The $p$-adic open string $N$-point tree amplitudes are
defined as
\begin{gather}
A^{(N)}\left(  \boldsymbol{k}\right)  =\label{Amplitude}\\%
{\displaystyle\int\limits_{\mathbb{Q}_{p}^{N-3}}}
{\displaystyle\prod\limits_{i=2}^{N-2}}
\left\vert x_{i}\right\vert _{p}^{\boldsymbol{k}_{1}\boldsymbol{k}_{i}%
}\left\vert 1-x_{i}\right\vert _{p}^{\boldsymbol{k}_{N-1}\boldsymbol{k}_{i}%
}\text{ }%
{\displaystyle\prod\limits_{2\leq i<j\leq N-2}}
\left\vert x_{i}-x_{j}\right\vert _{p}^{\boldsymbol{k}_{i}\boldsymbol{k}_{j}}%
{\displaystyle\prod\limits_{i=2}^{N-2}}
dx_{i}\text{,}\nonumber
\end{gather}
where $%
{\textstyle\prod\nolimits_{i=2}^{N-2}}
dx_{i}$ is the normalized Haar measure of $\mathbb{Q}_{p}^{N-3}$,
$\boldsymbol{k}=\left(  \boldsymbol{k}_{1},\ldots,\boldsymbol{k}_{N}\right)
$, $\boldsymbol{k}_{i}=\left(  k_{0,i},\ldots,k_{25,i}\right)  $,
$i=1,\ldots,N$, $N\geq4$, is the momentum vector of the $i$-th tachyon (with
Minkowski product $\boldsymbol{k}_{i}\boldsymbol{k}_{j}=-k_{0,i}%
k_{0,j}+k_{1,i}k_{1,j}+\cdots+k_{25,i}k_{25,j}$) obeying
\[
\sum_{i=1}^{N}\boldsymbol{k}_{i}=\boldsymbol{0}\text{, \ \ \ \ \ }%
\boldsymbol{k}_{i}\boldsymbol{k}_{i}=2\text{ \ for }i=1,\ldots,N.
\]
A central problem is to know whether or not integrals of type (\ref{Amplitude}%
) converge for some values $\boldsymbol{k}_{i}\boldsymbol{k}_{j}\in\mathbb{C}%
$. Our Theorem \ref{TheoremA} allows us to solve this problem. We take the
$p$-adic open string $N$-point tree integrals $Z^{(N)}(\boldsymbol{s})$ as
regularizations of the amplitudes $A^{(N)}\left(  \boldsymbol{k}\right)  $.
More precisely, we define
\[
A^{(N)}\left(  \boldsymbol{k}\right)  =Z^{(N)}(\boldsymbol{s})\mid
_{s_{ij}=\boldsymbol{k}_{i}\boldsymbol{k}_{j}}\text{with }i\in\left\{
1,\ldots,N-1\right\}  \text{, }j\in T\text{ or }i,j\in T,
\]
where $T=\left\{  2,\ldots,N-2\right\}  $. By Theorem \ref{TheoremA},
$A^{(N)}\left(  \boldsymbol{k}\right)  $ are well-defined rational functions
of the variables $p^{-\boldsymbol{k}_{i}\boldsymbol{k}_{j}}$, $i$,
$j\in\left\{  1,\ldots,N-1\right\}  $, which agree with integrals
(\ref{Amplitude}) when they converge. This definition allows us to recover all
the calculations made in \cite{B-F-O-W} and other similar publications. At
this point, it is relevant to mention that there is no similar result for the
Archimedean string amplitudes at the three level, as Witten pointed out \ in
\cite[p. 4]{Witten}.\ We notice that the string amplitudes $A^{(N)}\left(
\boldsymbol{k}\right)  $ are limits of local zeta functions when they are
considered as distributions. By a slight abuse of notation, this means that
\[
A^{(N)}\left(  \boldsymbol{k}\right)  =\lim_{r\rightarrow\infty}Z_{B_{r}%
^{N-3}}^{(N)}(\boldsymbol{k})\text{,}%
\]
for $\boldsymbol{k}$ in the natural domain of $Z^{(N)}(\boldsymbol{k})$.
Another important problem is to determine the existence of (in the sense of
quantum field theory) ultraviolet and infrared divergences for $A^{(N)}\left(
\boldsymbol{k}\right)  $. If we use the Euclidean product instead of the
Minkowski product to define $s_{ij}=\boldsymbol{k}_{i}\boldsymbol{k}_{j}$,
then $A^{(N)}\left(  \boldsymbol{k}\right)  $ has infrared divergences
($A^{(N)}\left(  0\right)  =+\infty$) and ultraviolet divergences
($A^{(N)}\left(  \boldsymbol{k}\right)  =+\infty$ for $\boldsymbol{k}%
_{i}\boldsymbol{k}_{j}>0$). The determination of the ultraviolet and infrared
divergences, in the sense of quantum field theory, in the signature
$-++\ldots+$ for $A^{(N)}\left(  \boldsymbol{k}\right)  $ is an open problem.
This problem requires the determination of the geometry of the natural domain
of function $Z^{(N)}(\boldsymbol{s})$. This type of problems has been not
studied in the case of multivariate local zeta functions.

Lerner and Missarov studied a class of $p$-adic integrals that includes
certain type\ of Feynman integrals and Koba-Nielsen amplitudes. They showed,
see \cite[Theorem 2]{Lerner}, \ that this type of integrals can be computed
recursively by using hierarchies, but they did not investigate the
convergence, or more generally the holomorphy, of the Koba-Nielsen amplitudes,
which is a delicate matter. On the other hand, the problem of regularization
string amplitudes has been recently considered by Witten in \cite{Witten}, by
using an analog of `the $i\varepsilon$ method'\ for regularizing Feynman
integrals. Our approach to the regularization of $p$-adic string amplitudes is
close to the technique of analytic regularization in quantum field theory,
see\ e.g. \cite{Speer}, \cite[Chapter 8]{Kleinert} and references therein.

In a more general framework, we point out that the string amplitudes at the
tree level and in general the Feynman amplitudes are `essentially' local zeta
functions (in the sense of Gel'fand, Sato, Weil, Bernstein,Tate, Igusa, Denef
and Loeser, among others), and thus, they are algebraic-geometric objects that
can be studied over several ground fields, for instance $\mathbb{R}$,
$\mathbb{C}$, $\mathbb{Q}_{p}$, $\mathbb{C}\left(  \left(  t\right)  \right)
$. Over each these fields these objects have similar mathematical properties.
As a consequence of our results and the theory of motivic Igusa zeta functions
due to Denef and Loeser \cite{DL1}-\cite{DL2}, a natural step is the
construction of motivic string amplitudes (motivic in the sense of motivic
integration) which specialized to the $p$-adic string amplitudes. In this
framework the limit $p\rightarrow1$ of Igusa's local zeta function makes
mathematical sense.

There is empirical evidence that $p$-adic strings are related to the ordinary
strings in the $p\rightarrow1$ limit, see e.g. \cite{Gerasimov}-\cite{Goshal},
\cite{Spokoiny:1988zk}, and the references therein. Denef and Loeser
established that the limit $p\rightarrow1$ of a Igusa's local zeta function
gives rise to an object, called topological zeta function \cite{DL3}. By using
Denef-Loeser's theory of topological zeta functions, in \ \cite{B-C-Zuniga},
we show that limit $p\rightarrow1$ of tree-level $p$-adic string amplitudes
give rise to certain amplitudes, that we have named \textit{Denef-Loeser
string amplitudes}. Gerasimov and Shatashvili showed that in limit
$p\rightarrow1$ the well-known non-local effective Lagrangian (reproducing the
tree-level $p$-adic string amplitudes) gives rise to a simple Lagrangian with
a logarithmic potential \cite{Gerasimov}. In \cite{B-C-Zuniga}, we show that
the Feynman amplitudes of this last Lagrangian are precisely the Denef-Loser amplitudes.

The article is organized as follows. In section \ref{Sect_p_adic_anal} we
present the basic aspects of the $p$-adic analysis needed in this article, and
in section \ref{Section3}, we prove the main result, Theorem \ref{TheoremA}.

\section{\label{Sect_p_adic_anal}Essential Ideas of $p$\textbf{-}Adic
Analysis}

In this section, we review some ideas and results on $p$-adic analysis that we
will use along this article. For an in-depth exposition, the reader may
consult \cite{Alberio et al}, \cite{Taibleson}, \cite{V-V-Z}.

\subsection{The field of $p$-adic numbers}

Throughout this article $p$ will denote a prime number. The field of $p-$adic
numbers $\mathbb{Q}_{p}$ is defined as the completion of the field of rational
numbers $\mathbb{Q}$ with respect to the $p-$adic norm $|\cdot|_{p} $, which
is defined as
\[
\left\vert x\right\vert _{p}=\left\{
\begin{array}
[c]{lll}%
0 & \text{if} & x=0\\
&  & \\
p^{-\gamma} & \text{if} & x=p^{\gamma}\frac{a}{b}\text{,}%
\end{array}
\right.
\]
where $a$ and $b$ are integers coprime with $p$. The integer $\gamma:=ord(x)
$, with $ord(0):=+\infty$, is called the\textit{\ }$p-$\textit{adic order of}
$x$. We extend the $p-$adic norm to $\mathbb{Q}_{p}^{n}$ by taking%
\[
||\boldsymbol{x}||_{p}:=\max_{1\leq i\leq n}|x_{i}|_{p},\qquad\text{for
}\boldsymbol{x}=(x_{1},\dots,x_{n})\in\mathbb{Q}_{p}^{n}.
\]
We define $ord(\boldsymbol{x})=\min_{1\leq i\leq n}\{ord(x_{i})\}$, then
$||\boldsymbol{x}||_{p}=p^{-ord(\boldsymbol{x})}$.\ The metric space $\left(
\mathbb{Q}_{p}^{n},||\cdot||_{p}\right)  $ is a complete ultrametric space. As
a topological space $\mathbb{Q}_{p}$\ is homeomorphic to a Cantor-like subset
of the real line, see e.g. \cite{Alberio et al}, \cite{V-V-Z}.

Any $p-$adic number $x\neq0$ has a unique expansion of the form
\[
x=p^{ord(x)}\sum_{i=0}^{\infty}x_{i}p^{i},
\]
where $x_{i}\in\{0,1,2,\dots,p-1\}$ and $x_{0}\neq0$.

For $r\in\mathbb{Z}$, denote by $B_{r}^{n}(\boldsymbol{a})=\{\boldsymbol{x}%
\in\mathbb{Q}_{p}^{n};||\boldsymbol{x}-\boldsymbol{a}||_{p}\leq p^{r}\}$
\textit{the ball of radius }$p^{r}$ \textit{with center at} $\boldsymbol{a}%
=(a_{1},\dots,a_{n})\in\mathbb{Q}_{p}^{n}$, and take $B_{r}^{n}(\boldsymbol{0}%
):=B_{r}^{n}$. Note that $B_{r}^{n}(\boldsymbol{a})=B_{r}(a_{1})\times
\cdots\times B_{r}(a_{n})$, where $B_{r}(a_{i}):=\{x\in\mathbb{Q}_{p}%
;|x_{i}-a_{i}|_{p}\leq p^{r}\}$ is the one-dimensional ball of radius $p^{r}$
with center at $a_{i}\in\mathbb{Q}_{p}$. The ball $B_{0}^{n}$ equals the
product of $n$ copies of $B_{0}=\mathbb{Z}_{p}$, \textit{the ring of }%
$p-$\textit{adic integers}. In addition, $B_{r}^{n}(\boldsymbol{a}%
)=\boldsymbol{a}+\left(  p^{-r}\mathbb{Z}_{p}\right)  ^{n}$. We also denote by
$S_{r}^{n}(\boldsymbol{a})=\{\boldsymbol{x}\in\mathbb{Q}_{p}^{n}%
;||\boldsymbol{x}-\boldsymbol{a}||_{p}=p^{r}\}$ \textit{the sphere of radius
}$p^{r}$ \textit{with center at} $\boldsymbol{a}\in\mathbb{Q}_{p}^{n}$, and
take $S_{r}^{n}(\boldsymbol{0}):=S_{r}^{n}$. We notice that $S_{0}%
^{1}=\mathbb{Z}_{p}^{\times}$ (the group of units of $\mathbb{Z}_{p}$), but
$\left(  \mathbb{Z}_{p}^{\times}\right)  ^{n}\subsetneq S_{0}^{n}$. The balls
and spheres are both open and closed subsets in $\mathbb{Q}_{p}^{n} $. In
addition, two balls in $\mathbb{Q}_{p}^{n}$ are either disjoint or one is
contained in the other.

As a topological space $\left(  \mathbb{Q}_{p}^{n},||\cdot||_{p}\right)  $ is
totally disconnected, i.e. the only connected \ subsets of $\mathbb{Q}_{p}%
^{n}$ are the empty set and the points. A subset of $\mathbb{Q}_{p}^{n}$ is
compact if and only if it is closed and bounded in $\mathbb{Q}_{p}^{n}$, see
e.g. \cite[Section 1.3]{V-V-Z}, or \cite[Section 1.8]{Alberio et al}. The
balls and spheres are compact subsets. Thus $\left(  \mathbb{Q}_{p}%
^{n},||\cdot||_{p}\right)  $ is a locally compact topological space.

\begin{remark}
There is a natural map, called the reduction $\operatorname{mod}p$ and denoted
as $\overline{\cdot}$, from $\mathbb{Z}_{p}$ onto $\mathbb{F}_{p}$, the finite
field with $p$ elements. More precisely, if $x=\sum_{j=0}^{\infty}x_{j}%
p^{j}\in\mathbb{Z}_{p}$, then $\overline{x}=\overline{x}_{0}\in\mathbb{F}%
_{p}=\left\{  \overline{0},\overline{1},\ldots,\overline{p-1}\right\}  $. If
$\boldsymbol{a}=(a_{1},\dots,a_{n})\in\mathbb{Z}_{p}^{n}$, then $\overline
{\boldsymbol{a}}=(\overline{a}_{1},\dots,\overline{a}_{n})$.
\end{remark}

\subsection{Integration on $\mathbb{Q}_{p}^{n}$}

Since $(\mathbb{Q}_{p},+)$ is a locally compact topological group, there
exists a Borel measure $dx$, called the Haar measure of $(\mathbb{Q}_{p},+)$,
unique up to multiplication by a positive constant, such that $\int_{U}dx>0 $
for every non-empty Borel open set $U\subset\mathbb{Q}_{p}$, and satisfying
$\int_{E+z}dx=\int_{E}dx$ for every Borel set $E\subset\mathbb{Q}_{p}$, see
e.g. \cite[Chapter XI]{Halmos}. If we normalize this measure by the condition
$\int_{\mathbb{Z}_{p}}dx=1$, then $dx$ is unique. From now on we denote by
$dx$ the normalized Haar measure of $(\mathbb{Q}_{p},+)$ and by $d^{n}%
\boldsymbol{x}$ the product measure on $(\mathbb{Q}_{p}^{n},+)$.

A function $\varphi:$ $\mathbb{Q}_{p}^{n}\rightarrow\mathbb{C}$ is said to be
\textit{locally constant} if for every $\boldsymbol{x}\in\mathbb{Q}_{p}^{n}$
there exists an open compact subset $U$, $\boldsymbol{x}\in U$, such that
$\varphi(\boldsymbol{x})=\varphi(\boldsymbol{u})$ for all $\boldsymbol{u}\in
U$. Any locally constant function $\varphi:$ $\mathbb{Q}_{p}^{n}%
\rightarrow\mathbb{C}$ can be expressed as a linear combination of
characteristic functions of the form $\varphi\left(  \boldsymbol{x}\right)
=\sum_{k=1}^{\infty}c_{k}{\LARGE 1}_{U_{k}}\left(  \boldsymbol{x}\right)  $,
where $c_{k}\in\mathbb{C}$ and ${\LARGE 1}_{U_{k}}\left(  \boldsymbol{x}%
\right)  $ is the characteristic function of $U_{k}$, an open compact subset
of $\mathbb{Q}_{p}^{n}$, for every $k$. If $\varphi$ has compact support, then
$\varphi\left(  \boldsymbol{x}\right)  =\sum_{k=1}^{L}c_{k}{\LARGE 1}_{U_{k}%
}\left(  \boldsymbol{x}\right)  $ and in this case
\[%
{\displaystyle\int\limits_{\mathbb{Q}_{p}^{n}}}
\varphi\left(  \boldsymbol{x}\right)  d^{n}\boldsymbol{x}=c_{1}%
{\displaystyle\int\limits_{U_{1}}}
d^{n}\boldsymbol{x}+\ldots+c_{L}%
{\displaystyle\int\limits_{U_{L}}}
d^{n}\boldsymbol{x}.
\]
A locally constant function with compact support is called a
\textit{Bruhat-Schwartz function}. These functions form a $\mathbb{C}$-vector
space denoted as $\mathcal{D}\left(  \mathbb{Q}_{p}^{n}\right)  $. By using
the fact that $\mathcal{D}\left(  \mathbb{Q}_{p}^{n}\right)  $ is a dense
subspace of $C_{c}\left(  \mathbb{Q}_{p}^{n}\right)  $, the $\mathbb{C}$-space
of continuous functions on $\mathbb{Q}_{p}^{n}$ with compact support, with the
topology of the uniform convergence, the functional \ $\varphi\rightarrow
\int_{\mathbb{Q}_{p}^{n}}\varphi\left(  \boldsymbol{x}\right)  d^{n}%
\boldsymbol{x}$, $\varphi\in\mathcal{D}\left(  \mathbb{Q}_{p}^{n}\right)  $
has a unique continuous extension to $C_{c}\left(  \mathbb{Q}_{p}^{n}\right)
$, as an unbounded linear functional. For integrating more general functions,
say locally integrable functions, \ the following notion of improper integral
will be used.

\begin{definition}
A function $\varphi\in L_{loc}^{1}$ is said to be integrable in $\mathbb{Q}%
_{p}^{n}$ if%
\[
\lim_{m\rightarrow+\infty}%
{\displaystyle\int\limits_{B_{m}^{n}}}
\varphi\left(  \boldsymbol{x}\right)  d^{n}\boldsymbol{x}=\lim_{m\rightarrow
+\infty}%
{\displaystyle\sum\limits_{j=-\infty}^{m}}
{\displaystyle\int\limits_{S_{j}^{n}}}
\varphi\left(  \boldsymbol{x}\right)  d^{n}\boldsymbol{x}%
\]
exists. If the limit exists, it is denoted as $%
{\textstyle\int\nolimits_{\mathbb{Q}_{p}^{n}}}
\varphi\left(  \boldsymbol{x}\right)  d^{n}\boldsymbol{x}$, and we say that
the\textbf{\ }(improper) integral exists.
\end{definition}

\subsection{Analytic change of variables}

A function $h:U\rightarrow\mathbb{Q}_{p}$ is said to be \textit{analytic} on
an open subset $U\subset\mathbb{Q}_{p}^{n}$, if for every $\boldsymbol{b}\in
U$ there exists an open subset $\widetilde{U}\subset U$, with $\boldsymbol{b}%
\in\widetilde{U}$, and a convergent power series $\sum_{i}a_{i}\left(
\boldsymbol{x}-\boldsymbol{b}\right)  ^{i}$ for $\boldsymbol{x}\in
\widetilde{U}$, such that $h\left(  \boldsymbol{x}\right)  =\sum
_{i\in\mathbb{N}^{n}}a_{i}\left(  \boldsymbol{x}-\boldsymbol{b}\right)  ^{i}$
for $\boldsymbol{x}\in\widetilde{U}$, with $\boldsymbol{x}^{i}=x_{1}^{i_{1}%
}\cdots x_{n}^{i_{n}}$, $\boldsymbol{i}=\left(  i_{1},\ldots,i_{n}\right)  $.
In this case, $\frac{\partial}{\partial x_{l}}h\left(  \boldsymbol{x}\right)
=\sum_{i\in\mathbb{N}^{n}}a_{i}\frac{\partial}{\partial x_{l}}\left(
\boldsymbol{x}-\boldsymbol{b}\right)  ^{i}$ is a convergent power series. Let
$U$, $V$ be open subsets of $\mathbb{Q}_{p}^{n}$. A mapping $\boldsymbol{h}%
:U\rightarrow V$, $\boldsymbol{h}=\left(  h_{1},\ldots,h_{n}\right)  $ is
called \textit{analytic} if each $h_{i}$ is analytic.

Let $\varphi:V$ $\rightarrow\mathbb{C}$ be a continuous function with compact
support, and let $\boldsymbol{h}:U\rightarrow V$ $\ $be an analytic mapping.
Then
\[%
{\textstyle\int\limits_{V}}
\varphi\left(  \boldsymbol{y}\right)  d^{n}\boldsymbol{y}=%
{\textstyle\int\limits_{U}}
\varphi\left(  \boldsymbol{h}(\boldsymbol{x})\right)  \left\vert
Jac(\boldsymbol{h}(\boldsymbol{x}))\right\vert _{p}d^{n}\boldsymbol{x}\text{,}%
\]
where $Jac(\boldsymbol{h}(\boldsymbol{z})):=\det\left[  \frac{\partial h_{i}%
}{\partial x_{j}}\left(  \boldsymbol{z}\right)  \right]  _{\substack{1\leq
i\leq n \\1\leq j\leq n}}$, see e.g. \cite[Section 10.1.2]{Bourbaki}.

\subsection{The multivariate Igusa zeta functions}

Let $f_{i}(\boldsymbol{x})\in\mathbb{Q}_{p}\left[  x_{1},\ldots,x_{n}\right]
$ be non-constant polynomials for $i=1,\ldots,l$, and \ let $\Phi$ be a
Bruhat-Schwartz function. The multivariate local zeta function attached to
$\left(  f_{1},\ldots,f_{l},\Phi\right)  $ (also called Igusa local zeta
function) is defined by the integral%
\[
Z_{\Phi}\left(  s_{1},\ldots,s_{l},;f_{1},\ldots,f_{l}\right)  =\int
\limits_{\mathbb{Q}_{p}^{n}\smallsetminus\cup_{i=1}^{l}f_{i}^{-1}%
(\boldsymbol{0})}\Phi\left(  \boldsymbol{x}\right)  \prod\limits_{i=1}%
^{l}\left\vert f_{i}(\boldsymbol{x})\right\vert _{p}^{s_{i}}d^{n}%
\boldsymbol{x}%
\]
for $\left(  s_{1},\ldots,s_{l}\right)  \in\mathbb{C}^{n}$ with
$\operatorname{Re}(s_{i})>0$, $i=1,\ldots,l$. This integral defines a
holomorphic function of $\left(  s_{1},\ldots,s_{l}\right)  $ in the
half-space $\operatorname{Re}(s_{i})>0$, $i=1,\ldots,l$. In the case $l=1,$
this assertion corresponds to Lemma 5.3.1 in \cite{Igusa}. For the general
case, we recall that a continuous complex-valued function defined in an open
set $A\subseteq$ $\mathbb{C}^{n}$, which is holomorphic in each variable
separately, is holomorphic in $A$. The multivariate local zeta functions admit
analytic continuations to the whole $\mathbb{C}^{n}$\ as rational functions of
the variables $p^{-s_{i}}$, $i=1,\ldots,l$, see \cite{Loeser}. The Igusa local
zeta functions are related with the number of solutions of polynomial
congruences $\operatorname{mod}$ $p^{m}$ and with exponential sums
$\operatorname{mod}$ $p^{m}$. There are many intriguing conjectures relating
the poles of local zeta functions with the topology of complex singularities,
see e.g. \cite{Denef}, \cite{Igusa}.

We want to highlight that the convergence of the local zeta functions depends
crucially on the fact that $\Phi$ has compact support. Consider the following
integral:%
\[
I(s)=%
{\displaystyle\int\limits_{\mathbb{Q}_{p}}}
\left\vert x\right\vert _{p}^{s}dx\text{, }s\in\mathbb{C}.
\]
Assume that $I(s_{0})$ exists for some $s_{0}\in\mathbb{R}$, then necessarily
the integrals%
\[
I_{0}(s_{0})=%
{\displaystyle\int\limits_{\mathbb{Z}_{p}}}
\left\vert x\right\vert _{p}^{s_{0}}dx\text{ \ and }I_{1}(s_{0})=%
{\displaystyle\int\limits_{\mathbb{Q}_{p}\smallsetminus\mathbb{Z}_{p}}}
\left\vert x\right\vert _{p}^{s_{0}}dx\text{ }%
\]
exist. The first integral is well-known, $I_{0}(s_{0})=\frac{1-p^{-1}%
}{1-p^{-1-s_{0}}}$ for $s_{0}>-1$. For the second integral, we use that
$\left\vert x\right\vert _{p}^{s_{0}}$ is locally integrable, and thus
\[
I_{1}(s_{0})=\sum\limits_{j=1}^{\infty}%
{\displaystyle\int\limits_{p^{-j}\mathbb{Z}_{p}^{\times}}}
\left\vert x\right\vert _{p}^{s_{0}}dx=\sum\limits_{j=1}^{\infty}p^{j+js_{0}}%
{\displaystyle\int\limits_{\mathbb{Z}_{p}^{\times}}}
dx=\left(  1-p^{-1}\right)  \sum\limits_{j=1}^{\infty}p^{j(1+s_{0})}%
<\infty\text{ }%
\]
if and only if $s_{0}<-1$. Then, integral $I(s)$ does not exist for any
$s\in\mathbb{R}$ and consequently $I(s)$ does not exist for any complex value
$s$.

For an in-depth discussion on local zeta functions the reader may consult
\cite{Denef}, \cite{Igusa0}-\cite{Igusa} and the references therein.

\section{\label{Section3}$p$-adic String Zeta Functions}

We fix an integer $N\geq4$. To each set $\left\{  i,j\right\}  $ with
$i,j\in\left\{  1,\ldots,N-1\right\}  $ we attach a complex variable $s_{ij}$.
We assume that the variables $s_{ij}$\ are algebraically independent. We set
$\ T:=\left\{  2,\ldots,N-2\right\}  $, then $D=2\left(  N-3\right)
+\frac{\left(  N-3\right)  \left(  N-4\right)  }{2}=\frac{N\left(  N-3\right)
}{2}$, and $\mathbb{C}^{D}$ is%
\[
\left\{
\begin{array}
[c]{lll}%
\left\{  s_{ij}\in\mathbb{C};i\in\left\{  1,N-1\right\}  ,j\in T\right\}  &
\text{if} & N=4\\
&  & \\
\left\{  s_{ij}\in\mathbb{C};i\in\left\{  1,N-1\right\}  ,j\in T\text{ or
}i,j\in T\text{ with }i<j\right\}  & \text{if} & N\geq5.
\end{array}
\right.
\]
We set$\ \boldsymbol{s}=\left(  s_{ij}\right)  \in\mathbb{C}^{D}$,
$\boldsymbol{x}=\left(  x_{2},\ldots,x_{N-2}\right)  \in\mathbb{Q}_{p}^{N-3}$,
and
\[
F\left(  \boldsymbol{s},\boldsymbol{x};N\right)  =%
{\displaystyle\prod\limits_{i=2}^{N-2}}
\left\vert x_{i}\right\vert _{p}^{s_{1i}}\left\vert 1-x_{i}\right\vert
_{p}^{s_{(N-1)i}}\text{ }%
{\displaystyle\prod\limits_{2\leq i<j\leq N-2}}
\left\vert x_{i}-x_{j}\right\vert _{p}^{s_{ij}}.
\]

\begin{definition}
The $p$\textit{-adic open string }$N$\textit{-point zeta function} is defined
as%
\begin{equation}
Z^{(N)}\left(  \boldsymbol{s}\right)  :=%
{\displaystyle\int\limits_{\mathbb{Q}_{p}^{N-3}\smallsetminus\Lambda}}
F\left(  \boldsymbol{s},\boldsymbol{x};N\right)
{\displaystyle\prod\limits_{i=2}^{N-2}}
dx_{i} \label{Zeta_1}%
\end{equation}
for $\boldsymbol{s}=\left(  s_{ij}\right)  \in\mathbb{C}^{D}$, where $\Lambda
$\ denotes the divisor%
\[
\left\{  \left(  x_{2},\ldots,x_{N-2}\right)  \in\mathbb{Q}_{p}^{N-3};%
{\displaystyle\prod\limits_{i=2}^{N-2}}
x_{i}\left(  1-x_{i}\right)  \text{ }%
{\displaystyle\prod\limits_{2\leq i<j\leq N-2}}
\left(  x_{i}-x_{j}\right)  =0\right\}  ,
\]
and $%
{\textstyle\prod\nolimits_{i=2}^{N-2}}
dx_{i}$ is the normalized Haar measure of $\mathbb{Q}_{p}^{N-3}$.
\end{definition}

\begin{remark}
We notice that the domain of integration in (\ref{Zeta_1}) is taken to be
$\mathbb{Q}_{p}^{N-3}\smallsetminus\Lambda$ in order to use $a^{s}=e^{s\ln a}%
$, with $a>0$ and $s\in\mathbb{C}$, as the definition of the complex power
function. The convergence of integral (\ref{Zeta_1}), as well as its
holomorphy, will be discussed later on.
\end{remark}

We define for $I\subseteq T$,\textit{\ the sector attached to }$I$ as
\[
Sect(I)=\left\{  \left(  x_{2},\ldots,x_{N-2}\right)  \in\mathbb{Q}_{p}%
^{N-3};\left\vert x_{i}\right\vert _{p}\leq1\text{ }\Leftrightarrow i\in
I\right\}
\]
and
\[
Z^{(N)}\left(  \boldsymbol{s};I\right)  =%
{\displaystyle\int\limits_{Sect(I)}}
F\left(  \boldsymbol{s},\boldsymbol{x};N\right)
{\displaystyle\prod\limits_{i=2}^{N-2}}
dx_{i}\text{.}%
\]
Hence%
\[
Z^{(N)}\left(  \boldsymbol{s}\right)  =\sum_{I\subseteq T}Z^{(N)}\left(
\boldsymbol{s};I\right)  .
\]

\begin{notation}
(i) The cardinality of a finite set $A$ will be denoted as $\left\vert
A\right\vert $.

\noindent(ii) We will use the symbol $\bigsqcup$ to denote the union of
disjoint sets.

\noindent(iii) Given a non-empty subset $I$ of $\left\{  2,\ldots,N-2\right\}
$ and $B$ a non-empty subset of $\mathbb{Q}_{p}$, we set%
\[
B^{\left\vert I\right\vert }=\left\{  \left(  x_{i}\right)  _{i\in I};x_{i}\in
B\right\}  .
\]
\noindent(iv) By convention, we define $%
{\textstyle\prod\nolimits_{i\in\varnothing}}
\cdot:=1$, $\sum_{i\in\varnothing}\cdot:=0$, and if $J=\varnothing$, then
$\int_{B^{\left\vert J\right\vert }}\cdot:=1$.

\noindent(v) The indices $i$, $j$ will run over subsets of $T$, if we do not
specify any subset, we will assume that is $T$.
\end{notation}

\begin{lemma}
\label{Lema_A}With the above notation the following formulas hold: (i)
$F\left(  \boldsymbol{s},\boldsymbol{x};N\right)  \mid_{Sect(I)}%
\allowbreak=F_{0}\left(  \boldsymbol{s},\boldsymbol{x};N\right)  F_{1}\left(
\boldsymbol{s},\boldsymbol{x};N\right)  $, where%
\[
F_{0}\left(  \boldsymbol{s},\boldsymbol{x};N\right)  :=%
{\displaystyle\prod\limits_{i\in I}}
\left\vert x_{i}\right\vert _{p}^{s_{1i}}\left\vert 1-x_{i}\right\vert
_{p}^{s_{(N-1)i}}\text{ }%
{\displaystyle\prod\limits_{\substack{2\leq i<j\leq N-2 \\i,j\in I}}}
\left\vert x_{i}-x_{j}\right\vert _{p}^{s_{i}{}_{j}}%
\]
and
\[
F_{1}\left(  \boldsymbol{s},\boldsymbol{x};N\right)  :=\prod\limits_{i\in
T\smallsetminus I}\left\vert x_{i}\right\vert _{p}^{s_{1}{}_{i}+s_{(N-1)i}%
+\sum_{\substack{2\leq j\leq N-2 \\j\neq i,\text{ }j\in I}}s_{i}{}_{j}}\text{
}%
{\displaystyle\prod\limits_{\substack{2\leq i<j\leq N-2 \\i,j\in
T\smallsetminus I}}}
\left\vert x_{i}-x_{j}\right\vert _{p}^{s_{i}{}_{j}}.
\]

\noindent(ii) If $\operatorname{Re}\left(  s_{1i}\right)  +\operatorname{Re}%
\left(  s_{\left(  N-1\right)  i}\right)  +\sum_{2\leq j\leq N-2,j\neq
i}\operatorname{Re}\left(  s_{ij}\right)  +1<0$ for $i\in T\smallsetminus I$,
and $\operatorname{Re}\left(  s_{ij}\right)  >-1$ for $i,j\in T\smallsetminus
I$, then
\begin{align*}
&
{\displaystyle\int\limits_{\left(  \mathbb{Q}_{p}\smallsetminus\mathbb{Z}%
_{p}\right)  ^{\left\vert T\smallsetminus I\right\vert }}}
F_{1}\left(  \boldsymbol{s},\boldsymbol{x};N\right)
{\displaystyle\prod\limits_{i\in T\smallsetminus I}}
dx_{i}\\
&  =p^{M(\boldsymbol{s})}\text{ }%
{\displaystyle\int\limits_{\mathbb{Z}_{p}^{\left\vert T\smallsetminus
I\right\vert }}}
\frac{%
{\displaystyle\prod\limits_{\substack{2\leq i<j\leq N-2 \\i,j\in
T\smallsetminus I}}}
\left\vert y_{i}-y_{j}\right\vert _{p}^{s_{ij}}}{%
{\displaystyle\prod\limits_{i\in T\smallsetminus I}}
\left\vert y_{i}\right\vert _{p}^{2+s_{1i}+s_{\left(  N-1\right)  i}%
+\sum_{2\leq j\leq N-2,j\neq i}s_{ij}}}%
{\displaystyle\prod\limits_{i\in T\smallsetminus I}}
dy_{i},
\end{align*}
where $M(\boldsymbol{s}):=\left\vert T\smallsetminus I\right\vert +\sum_{i\in
T\smallsetminus I}(s_{1i}+s_{\left(  N-1\right)  i})+\sum_{\substack{2\leq
i<j\leq N-2 \\i\in T\smallsetminus I,j\in T}}s_{ij}+\sum_{\substack{2\leq
i<j\leq N-2 \\i\in I,j\in T\smallsetminus I }}s_{ij}$.

\noindent(iii) If $\operatorname{Re}\left(  s_{1i}\right)  +\operatorname{Re}%
\left(  s_{\left(  N-1\right)  i}\right)  +\sum_{2\leq j\leq N-2,j\neq
i}\operatorname{Re}\left(  s_{ij}\right)  +1<0$ for $i\in T\smallsetminus I$,
$\operatorname{Re}\left(  s_{ij}\right)  >-1$ for $i,j\in T\smallsetminus I$,
$\operatorname{Re}\left(  s_{1i}\right)  >-1$ for $i\in I$ and
$\operatorname{Re}\left(  s_{(N-1)i}\right)  >-1$\ for $i\in I$, then
\begin{gather*}
Z^{(N)}\left(  \boldsymbol{s};I\right)  =p^{M(\boldsymbol{s})}\left\{
{\displaystyle\int\limits_{\mathbb{Z}_{p}^{\left\vert I\right\vert }}}
F_{0}\left(  \boldsymbol{s},\boldsymbol{x};N\right)
{\displaystyle\prod\limits_{i\in I}}
dx_{i}\right\} \\
\times\left\{
{\displaystyle\int\limits_{\mathbb{Z}_{p}^{\left\vert T\smallsetminus
I\right\vert }}}
\frac{%
{\displaystyle\prod\limits_{\substack{2\leq i<j\leq N-2 \\i,j\in
T\smallsetminus I}}}
\left\vert x_{i}-x_{j}\right\vert _{p}^{s_{ij}}}{%
{\displaystyle\prod\limits_{i\in T\smallsetminus I}}
\left\vert x_{i}\right\vert _{p}^{2+s_{1i}+s_{\left(  N-1\right)  i}%
+\sum_{2\leq j\leq N-2,j\neq i}s_{ij}}}%
{\displaystyle\prod\limits_{i\in T\smallsetminus I}}
dx_{i}\right\} \\
=:p^{M(\boldsymbol{s})}Z_{0}^{(N)}\left(  \boldsymbol{s};I\right)  Z_{1}%
^{(N)}\left(  \boldsymbol{s};T\smallsetminus I\right)  .
\end{gather*}

\end{lemma}

\begin{remark}
Later on we will show that the integrals in the right-hand side in the
formulas given in (ii) and (iii) are convergent and holomorphic functions on a
certain subset of $\mathbb{C}^{D}$ for all $I\subseteq T$.
\end{remark}

\begin{proof}
(i) Notice that $F\left(  \boldsymbol{s},\boldsymbol{x};N\right)
\mid_{Sect(I)}$ equals
\begin{gather}%
{\displaystyle\prod\limits_{i\in I}}
\left\vert x_{i}\right\vert _{p}^{s_{1}{}_{i}}\left\vert 1-x_{i}\right\vert
_{p}^{s_{\left(  N-1\right)  }{}_{i}}\text{ }%
{\displaystyle\prod\limits_{i\in T\smallsetminus I}}
\left\vert x_{i}\right\vert _{p}^{s_{1}{}_{i}+s_{(N-1)i}}\text{ }%
{\displaystyle\prod\limits_{\substack{2\leq i<j\leq N-2 \\i,j\in I}}}
\left\vert x_{i}-x_{j}\right\vert _{p}^{s_{i}{}_{j}}\text{ }\times\nonumber\\%
{\displaystyle\prod\limits_{\substack{2\leq i<j\leq N-2 \\i,j\in
T\smallsetminus I}}}
\left\vert x_{i}-x_{j}\right\vert _{p}^{s_{i}{}_{j}}\text{ }%
{\displaystyle\prod\limits_{\substack{2\leq i<j\leq N-2 \\i\in T\smallsetminus
I,j\in I}}}
\left\vert x_{i}\right\vert _{p}^{s_{i}{}_{j}}\text{ }%
{\displaystyle\prod\limits_{\substack{2\leq i<j\leq N-2 \\i\in I,j\in
T\smallsetminus I}}}
\left\vert x_{j}\right\vert _{p}^{s_{i}{}_{j}}. \label{Formula_A}%
\end{gather}
Now, by using that $s_{i}{}_{j}=s_{ji}$,%
\begin{gather}%
{\displaystyle\prod\limits_{\substack{2\leq i<j\leq N-2 \\i\in T\smallsetminus
I,j\in I }}}
\left\vert x_{i}\right\vert _{p}^{s_{i}{}_{j}}\text{ }%
{\displaystyle\prod\limits_{\substack{2\leq i<j\leq N-2 \\i\in I,j\in
T\smallsetminus I}}}
\left\vert x_{j}\right\vert _{p}^{s_{i}{}_{j}}=%
{\displaystyle\prod\limits_{\substack{2\leq i<j\leq N-2 \\i\in T\smallsetminus
I,j\in I}}}
\left\vert x_{i}\right\vert _{p}^{s_{i}{}_{j}}\text{ }%
{\displaystyle\prod\limits_{\substack{2\leq j<i\leq N-2 \\j\in I,i\in
T\smallsetminus I}}}
\left\vert x_{i}\right\vert _{p}^{s_{i}{}_{j}}\nonumber\\
=%
{\displaystyle\prod\limits_{\substack{2\leq j,i\leq N-2 \\i\neq j,\text{ }i\in
T\smallsetminus I,j\in I}}}
\left\vert x_{i}\right\vert _{p}^{s_{i}{}_{j}}=%
{\displaystyle\prod\limits_{i\in T\smallsetminus I}}
\left\vert x_{i}\right\vert _{p}^{\sum_{\substack{2\leq j\leq N-2 \\j\neq
i,\text{ }j\in I}}s_{i}{}_{j}}. \label{Formula_B}%
\end{gather}
The announced formula follows from (\ref{Formula_A})-(\ref{Formula_B}).

(ii) For $\left\vert T\smallsetminus I\right\vert \geq1$, we set
\[
J\left(  \boldsymbol{s};T\smallsetminus I\right)  :=%
{\displaystyle\int\limits_{\left(  \mathbb{Q}_{p}\smallsetminus\mathbb{Z}%
_{p}\right)  ^{\left\vert T\smallsetminus I\right\vert }}}
F_{1}\left(  \boldsymbol{s},\boldsymbol{x};N\right)
{\displaystyle\prod\limits_{i\in T\smallsetminus I}}
dx_{i},
\]
and for $l\in\mathbb{N\smallsetminus}\left\{  0\right\}  $,%
\[
\left(  \mathbb{Q}_{p}\smallsetminus\mathbb{Z}_{p}\right)  _{-l}^{\left\vert
T\smallsetminus I\right\vert }:=\left\{  \left(  x_{i}\right)  _{i\in
T\smallsetminus I}\in\left(  \mathbb{Q}_{p}\smallsetminus\mathbb{Z}%
_{p}\right)  ^{\left\vert T\smallsetminus I\right\vert };-l\leq ord(x_{i}%
)\leq-1\text{ for }i\in T\smallsetminus I\right\}  ,
\]%
\[
\left(  p\mathbb{Z}_{p}\right)  _{l}^{\left\vert T\smallsetminus I\right\vert
}:=\left\{  \left(  x_{i}\right)  _{i\in T\smallsetminus I}\in\left(
p\mathbb{Z}_{p}\right)  ^{\left\vert T\smallsetminus I\right\vert };1\leq
ord(x_{i})\leq l\text{ for }i\in T\smallsetminus I\right\}  ,
\]
and
\[
J_{-l}\left(  \boldsymbol{s};T\smallsetminus I\right)  :=%
{\displaystyle\int\limits_{\left(  \mathbb{Q}_{p}\smallsetminus\mathbb{Z}%
_{p}\right)  _{-l}^{\left\vert T\smallsetminus I\right\vert }}}
F_{1}\left(  \boldsymbol{s},\boldsymbol{x};N\right)
{\displaystyle\prod\limits_{i\in T\smallsetminus I}}
dx_{i}.
\]
Notice that $\left(  \mathbb{Q}_{p}\smallsetminus\mathbb{Z}_{p}\right)
_{-l}^{\left\vert T\smallsetminus I\right\vert }$, $\left(  p\mathbb{Z}%
_{p}\right)  _{l}^{\left\vert T\smallsetminus I\right\vert }$ are compact sets
and that
\[%
\begin{array}
[c]{ccc}%
\left(  \mathbb{Q}_{p}\smallsetminus\mathbb{Z}_{p}\right)  _{-l}^{\left\vert
T\smallsetminus I\right\vert } & \rightarrow & \left(  p\mathbb{Z}_{p}\right)
_{l}^{\left\vert T\smallsetminus I\right\vert }\\
&  & \\
\left(  x_{i}\right)  _{i\in T\smallsetminus I} & \rightarrow & \left(
\sigma\left(  x_{i}\right)  \right)  _{i\in T\smallsetminus I},
\end{array}
\]
with $\sigma\left(  x_{i}\right)  =\frac{1}{y_{i}}$ is an analytic change of
variables satisfying $%
{\textstyle\prod\nolimits_{i\in T\smallsetminus I}}
dx_{i}=$ $\allowbreak%
{\textstyle\prod\nolimits_{i\in T\smallsetminus I}}
\frac{dy_{i}}{\left\vert y_{i}\right\vert _{p}^{2}}$, then by using this
change of variables and the fact that
\begin{align*}
&
{\displaystyle\prod\limits_{i\in T\smallsetminus I}}
\left\vert y_{i}\right\vert _{p}^{s_{1i}+s_{\left(  N-1\right)  i}%
+\sum_{\substack{2\leq j\leq N-2 \\j\neq i,j\in I}}s_{ij}}%
{\displaystyle\prod\limits_{\substack{2\leq i<j\leq N-2 \\i\in T\smallsetminus
I,j\in T\smallsetminus I}}}
\left\vert y_{i}\right\vert _{p}^{s_{i}{}_{j}}\text{ }%
{\displaystyle\prod\limits_{\substack{2\leq i<j\leq N-2 \\i\in T\smallsetminus
I,j\in T\smallsetminus I}}}
\left\vert y_{j}\right\vert _{p}^{s_{i}{}_{j}}\\
&  =%
{\displaystyle\prod\limits_{i\in T\smallsetminus I}}
\left\vert y_{i}\right\vert _{p}^{s_{1i}+s_{\left(  N-1\right)  i}%
+\sum_{\substack{2\leq j\leq N-2 \\j\neq i,j\in I}}s_{ij}}%
{\displaystyle\prod\limits_{i\in T\smallsetminus I}}
\left\vert y_{i}\right\vert _{p}^{\sum_{\substack{2\leq j\leq N-2 \\j\neq
i,j\in T\smallsetminus I}}s_{ij}}\\
&  =%
{\displaystyle\prod\limits_{i\in T\smallsetminus I}}
\left\vert y_{i}\right\vert _{p}^{s_{1i}+s_{\left(  N-1\right)  i}+\sum_{2\leq
j\leq N-2,j\neq i}s_{ij}},
\end{align*}
we have
\begin{equation}
J_{-l}\left(  \boldsymbol{s};T\smallsetminus I\right)  =%
{\displaystyle\int\limits_{\left(  p\mathbb{Z}_{p}\right)  _{l}^{\left\vert
T\smallsetminus I\right\vert }}}
\frac{%
{\displaystyle\prod\limits_{\substack{2\leq i<j\leq N-2 \\i,j\in
T\smallsetminus I}}}
\left\vert y_{i}-y_{j}\right\vert _{p}^{s_{i}{}_{j}}%
{\displaystyle\prod\limits_{i\in T\smallsetminus I}}
dy_{i}}{%
{\displaystyle\prod\limits_{i\in T\smallsetminus I}}
\left\vert y_{i}\right\vert _{p}^{s_{1i}+s_{\left(  N-1\right)  i}+\sum_{2\leq
j\leq N-2,j\neq i}s_{ij}+2}}. \label{I_l}%
\end{equation}
Then $\lim_{l\rightarrow\infty}J_{-l}\left(  \boldsymbol{s};T\smallsetminus
I\right)  =J\left(  \boldsymbol{s};T\smallsetminus I\right)  $. Indeed, the
formula follows from the dominated convergence theorem, by using that
$\left\vert y_{i}-y_{j}\right\vert _{p}^{\operatorname{Re}(s_{i}{}_{j})}<1$
for $y_{i} $, $y_{j}\in p\mathbb{Z}_{p}$, and the fact that $\int
_{p\mathbb{Z}_{p}}\frac{1}{\left\vert y\right\vert _{p}^{s}}dy$ converges for
$\operatorname{Re}\left(  s\right)  <1$. Finally, the announced formula
follows from (\ref{I_l}) by a change of variables.

(iii) It is a consequence of (i)-(ii).
\end{proof}

\subsection{ \label{Nota_Lemma_A}The road map of the proof of the main result}

From Lemma \ref{Lema_A}, we have
\begin{equation}
Z^{(N)}\left(  \boldsymbol{s}\right)  =\sum_{I\subseteq T}p^{M(\boldsymbol{s}%
)}Z_{0}^{(N)}\left(  \boldsymbol{s};I\right)  Z_{1}^{(N)}\left(
\boldsymbol{s};T\smallsetminus I\right)  . \label{Formula_Zeta_amplirtude}%
\end{equation}
From now on, for the sake of simplicity we will use the notation
$Z^{(N)}\left(  \boldsymbol{s}\right)  =Z\left(  \boldsymbol{s}\right)  $,
$Z_{0}^{(N)}\left(  \boldsymbol{s};I\right)  =Z_{0}\left(  \boldsymbol{s}%
;I\right)  $, and $Z_{1}^{(N)}\left(  \boldsymbol{s};T\smallsetminus I\right)
=Z_{1}\left(  \boldsymbol{s};T\smallsetminus I\right)  $. By convention
$Z_{0}\left(  \boldsymbol{s};\varnothing\right)  =1$, $Z_{1}\left(
\boldsymbol{s};\varnothing\right)  =1$. A central goal of this article is to
show that $Z\left(  \boldsymbol{s}\right)  $ has an analytic continuation to
the whole $\mathbb{C}^{D}$ as a rational function in the variables
$p^{-s_{ij}}$. To establish this result, we show that all functions appearing
on the right-hand side of formula (\ref{Formula_Zeta_amplirtude}) admit
analytic continuations to the whole $\mathbb{C}^{D}$ as rational functions in
the variables $p^{-s_{ij}}$, and that each of these functions is holomorphic
on certain domain, and that the intersection of all these domains contains an
open and connected subset of $\mathbb{C}^{D}$, which allows us to use the
principle of analytic continuation.

Propositions \ref{Proposition_1} and \ref{Proposition_2} are dedicated to
showing that $Z_{1}(\boldsymbol{s};J)$ and $Z_{0}(\boldsymbol{s};I)$,
respectively, admit analytic continuations as rational functions of the
$p^{-s_{ij}}$; the main theorem then follows from these, after an analysis of
the support conditions (given in Lemmas \ref{Lemma_3A} and \ref{Lemma_H_0},
respectively). The proof of the propositions relies on a recursive expression
for the $Z_{i}(\boldsymbol{s};J)$s in terms of auxiliary functions
$L_{0}(\boldsymbol{s};I)$, $L_{1}(\boldsymbol{s};I,J)$, and $L_{2}%
(\boldsymbol{s};I,J,K)$. Lemmas \ref{Lemma_L_1} and \ref{Lema_L_2} demonstrate
the form of the respective analytic continuations of $L_{1}(\boldsymbol{s}%
;I,J)$ and $L_{2}(\boldsymbol{s};I,J,K)$; Lemma \ref{Lemma_L_0} computes
$L_{0}(\boldsymbol{s};I)$ in terms of $L_{1}(\boldsymbol{s};I,J)$. (A
degeneration formula, relating $L_{1}(\boldsymbol{s};I,J)$ to $L_{2}%
(\boldsymbol{s};I,J,K)$, appears as Remark \ref{Nota_Lemma_L_2}). Lemmas
\ref{Lemma_F} and \ref{Lemma_M_a_J} address the computation of other auxiliary
functions, $Z(s_{1},s_{2},s_{3})$ and $M_{1}(\boldsymbol{s};J)$; Lemma
\ref{Lemma_partition} is a preliminary result regarding indexing.

A natural question is to know if the auxiliary integrals introduced to compute
$Z(\boldsymbol{s})$ have a natural interpretation in the context of $p$-adic
string theory. In our opinion, these a auxiliary integrals are merely
organizational devices, since they are defined over particular regions in the
moduli space.

\subsection{Some $p$-adic integrals}

We compute some $p$-adic integrals needed for calculating $Z_{0}\left(
\boldsymbol{s};I\right)  $ and $Z_{1}\left(  \boldsymbol{s};I\right)  $.

Let $J$ be a subset of $T$ with $\left\vert J\right\vert \geq2$. We define%
\begin{equation}
L_{0}\left(  \left(  s_{ij}\right)  _{_{\substack{2\leq i<j\leq N-2\\i,j\in
J}}};J\right)  :=L_{0}\left(  \boldsymbol{s};J\right)  =%
{\displaystyle\int\limits_{(\mathbb{Z}_{p}^{\mathbb{\times}})^{\left\vert
J\right\vert }}}
{\displaystyle\prod\limits_{_{\substack{2\leq i<j\leq N-2\\i,j\in J}}}}
\left\vert x_{i}-x_{j}\right\vert _{p}^{s_{i}{}_{j}}%
{\displaystyle\prod\limits_{i\in J}}
dx_{i}\text{ } \label{L_0_N}%
\end{equation}
for $\operatorname{Re}\left(  s_{ij}\right)  >0$ for any $ij$, and
\begin{equation}
L_{1}\left(  \left(  s_{ij}\right)  _{_{\substack{2\leq i<j\leq N-2\\i,j\in
J}}};J,K\right)  :=L_{1}\left(  \boldsymbol{s};J,K\right)  =%
{\displaystyle\int\limits_{\mathbb{Z}_{p}^{\left\vert J\right\vert }}}
{\displaystyle\prod\limits_{\substack{\left(  i,j\right)  \in K}}}
\left\vert x_{i}-x_{j}\right\vert _{p}^{s_{ij}}%
{\displaystyle\prod\limits_{i\in J}}
dx_{i}, \label{L_1_N}%
\end{equation}
where $K\subseteq T_{J}:=\left\{  \left(  i,j\right)  \in T\times T;2\leq
i<j\leq N-2,i,j\in J\right\}  $ and $\operatorname{Re}\left(  s_{ij}\right)
>0$ for any $ij$. Notice that if $\left\vert J\right\vert =1$, then
$L_{0}\left(  \boldsymbol{s};J\right)  =1-p^{-1}$ and $K=\varnothing$ which
implies $L_{1}\left(  \boldsymbol{s};J,K\right)  =1$. A precise definition of
integrals $L_{0}\left(  \boldsymbol{s};J\right)  $ requires to integrate on
\[
(\mathbb{Z}_{p}^{\mathbb{\times}})^{\left\vert J\right\vert }\smallsetminus
\left\{  x\in(\mathbb{Z}_{p}^{\mathbb{\times}})^{\left\vert J\right\vert };%
{\displaystyle\prod\limits_{\substack{2\leq i<j\leq N-2\\i,j\in J}}}
\left(  x_{i}-x_{j}\right)  =0\right\}  .
\]
A similar \ consideration is required for $L_{1}\left(  \boldsymbol{s}%
;J,K\right)  $. However, for the sake of simplicity we use definitions
(\ref{L_0_N})-(\ref{L_1_N}). We will use this simplified notation later on for
similar integrals. The integrals $L_{0}\left(  \boldsymbol{s};J\right)  $,
$L_{1}\left(  \boldsymbol{s};J,K\right)  $ are $p$-adic multivariate local
zeta function, these functions were studied by Loeser in \cite{Loeser}. In
particular, it is known that these functions have an analytic continuation to
$\mathbb{C}^{D}$\ as rational functions in the variables $p^{-s_{ij}}$ and
that they are holomorphic functions on $\operatorname{Re}\left(
s_{ij}\right)  >0$ for any $ij$.

\begin{remark}
\label{Nota_basic_const}Let $J$ be subset of $T$, with $\left\vert
J\right\vert \geq2$. Set
\[
T_{J}=\left\{  \left(  i,j\right)  \in T\times T;2\leq i<j\leq N-2,i,j\in
J\right\}
\]
as before. For $\overline{\boldsymbol{a}}=\left(  \overline{a}_{i}\right)
_{i\in J}\in(\mathbb{F}_{p}^{\mathbb{\times}})^{\left\vert J\right\vert
}\smallsetminus\overline{\Delta}\left(  J\right)  $, with
\[
\overline{\Delta}\left(  J\right)  :=\left\{  \overline{\boldsymbol{a}}%
\in(\mathbb{F}_{p}^{\mathbb{\times}})^{\left\vert J\right\vert };\overline
{a}_{i}\neq\overline{a}_{j}\text{ for }i\neq j\text{, with }i,j\in J\right\}
,
\]
we set%
\[
K(\overline{\boldsymbol{a}}):=\left\{  \left(  i,j\right)  \in T_{J}%
;\overline{a}_{i}=\overline{a}_{j}\right\}  .
\]
Now, we introduce on $(\mathbb{F}_{p}^{\mathbb{\times}})^{\left\vert
J\right\vert }\smallsetminus\overline{\Delta}\left(  J\right)  $, the
following equivalence relation:%
\[
\overline{\boldsymbol{a}}\sim\overline{\boldsymbol{b}}\text{ }\Leftrightarrow
K(\overline{\boldsymbol{a}})=K(\overline{\boldsymbol{b}}).
\]
We denote by $\overline{A}(\overline{\boldsymbol{a}})=\left\{  \overline
{\boldsymbol{b}}\in(\mathbb{F}_{p}^{\mathbb{\times}})^{\left\vert J\right\vert
}\smallsetminus\overline{\Delta}\left(  J\right)  ;\overline{\boldsymbol{a}%
}\sim\overline{\boldsymbol{b}}\right\}  $, the equivalence class defined by
$\overline{\boldsymbol{a}}\in(\mathbb{F}_{p}^{\mathbb{\times}})^{\left\vert
J\right\vert }\smallsetminus\overline{\Delta}\left(  J\right)  $. For
instance, if $\ \overline{\boldsymbol{a}}=\overline{\boldsymbol{1}}=\left(
\overline{1}\right)  _{i\in J}$, then $\overline{A}\left(  \overline
{\boldsymbol{1}}\right)  =%
{\textstyle\bigsqcup\nolimits_{\overline{b}\in\mathbb{F}_{p}^{\mathbb{\times}%
}}}
\left\{  \overline{b}\left(  \overline{1}\right)  _{i\in J}\right\}  $. By
taking a unique representative in each equivalence class, we obtain
$\mathcal{R}(J)\subset(\mathbb{F}_{p}^{\mathbb{\times}})^{\left\vert
J\right\vert }\smallsetminus\overline{\Delta}\left(  J\right)  $ such that
\[
(\mathbb{F}_{p}^{\mathbb{\times}})^{\left\vert J\right\vert }=\bigsqcup
\limits_{\overline{\boldsymbol{a}}\in R(J)}\overline{A}(\overline
{\boldsymbol{a}})\bigsqcup\overline{\Delta}\left(  J\right)  .
\]
Given a subset $K\subseteq T_{J}$ with $K=\left\{  \left(  i_{1},j_{1}\right)
,\ldots,\left(  i_{m},j_{m}\right)  \right\}  $, we define $K_{\text{list}%
}=\left\{  i_{1},j_{1},\ldots,i_{m},j_{m}\right\}  \subset J$. We will use the
notation $K_{\text{list}}(\overline{\boldsymbol{a}})$ to mean $K(\overline
{\boldsymbol{a}})_{\text{list}}$, for $\overline{\boldsymbol{a}}\in
(\mathbb{F}_{p}^{\mathbb{\times}})^{\left\vert J\right\vert }$. Notice that
$K(\overline{\boldsymbol{a}})\subset K_{\text{list}}(\overline{\boldsymbol{a}%
})\times K_{\text{list}}(\overline{\boldsymbol{a}})$, $\left\vert
K_{\text{list}}(\overline{\boldsymbol{a}})\right\vert \geq2$ for any
$\overline{\boldsymbol{a}}\in(\mathbb{F}_{p}^{\mathbb{\times}})^{\left\vert
J\right\vert }\smallsetminus\overline{\Delta}\left(  J\right)  $ and that
$K_{\text{list}}(\overline{\boldsymbol{1}})=J$.
\end{remark}

\begin{lemma}
\label{Lemma_L_0}If $\left\vert J\right\vert \geq2$, then, with the notation
of Remark \ref{Nota_basic_const}, the following formula holds:%
\[
L_{0}\left(  \boldsymbol{s};J\right)  =\sum\limits_{\overline{\boldsymbol{a}%
}\in\mathcal{R}(J)}\left\vert \overline{A}(\overline{\boldsymbol{a}%
})\right\vert p^{-\left\vert J\right\vert -\sum_{_{\left(  i,j\right)  \in
K(\overline{\boldsymbol{a}})}}s_{ij}}L_{1}\left(  \boldsymbol{s}%
;K_{\text{list}}(\overline{\boldsymbol{a}}),K(\overline{\boldsymbol{a}%
})\right)  +\left\vert \overline{\Delta}\left(  J\right)  \right\vert
p^{-\left\vert J\right\vert }%
\]
for $\operatorname{Re}(s_{ij})>0$ for all $i$, $j\in J$.
\end{lemma}

\begin{proof}
For $\overline{\boldsymbol{a}}\in(\mathbb{F}_{p}^{\mathbb{\times}%
})^{\left\vert J\right\vert }\smallsetminus\overline{\Delta}\left(  J\right)
$, set $A(\overline{\boldsymbol{a}}):=\left\{  \boldsymbol{b}+p\boldsymbol{x}%
;\overline{\boldsymbol{b}}\in\overline{A}\left(  \overline{\boldsymbol{a}%
}\right)  \right\}  $, and for $\overline{\Delta}\left(  J\right)  $,
$\Delta\left(  J\right)  :=\left\{  \boldsymbol{a}+p\boldsymbol{x}%
;\overline{\boldsymbol{a}}\in\overline{\Delta}\left(  J\right)  \right\}  $.
Now
\begin{multline*}
L_{0}\left(  \boldsymbol{s};J\right)  =\sum\limits_{\overline{\boldsymbol{a}%
}\in(\mathbb{F}_{p}^{\mathbb{\times}})^{\left\vert J\right\vert }}%
{\displaystyle\int\limits_{\boldsymbol{a}+\left(  p\mathbb{Z}_{p}\right)
^{\left\vert J\right\vert }}}
{\displaystyle\prod\limits_{\substack{2\leq i<j\leq N-2 \\i,j\in J}}}
\left\vert x_{i}-x_{j}\right\vert _{p}^{s_{ij}}%
{\displaystyle\prod\limits_{i\in J}}
dx_{i}\\
=\sum\limits_{\overline{\boldsymbol{a}}\in\mathcal{R}(J)}\sum
\limits_{\overline{\boldsymbol{b}}\in\overline{A}(\overline{\boldsymbol{a}})}%
{\displaystyle\int\limits_{\boldsymbol{b}+\left(  p\mathbb{Z}_{p}\right)
^{\left\vert J\right\vert }}}
{\displaystyle\prod\limits_{\substack{2\leq i<j\leq N-2 \\i,j\in J}}}
\left\vert x_{i}-x_{j}\right\vert _{p}^{s_{ij}}%
{\displaystyle\prod\limits_{i\in J}}
dx_{i}+\\
\sum\limits_{\overline{\boldsymbol{a}}\in\overline{\Delta}\left(  J\right)  }%
{\displaystyle\int\limits_{\boldsymbol{a}+\left(  p\mathbb{Z}_{p}\right)
^{\left\vert J\right\vert }}}
{\displaystyle\prod\limits_{\substack{2\leq i<j\leq N-2 \\i,j\in J}}}
\left\vert x_{i}-x_{j}\right\vert _{p}^{s_{ij}}%
{\displaystyle\prod\limits_{i\in J}}
dx_{i}\\
=\sum\limits_{\overline{\boldsymbol{a}}\in\mathcal{R}(J)}\left\vert
\overline{A}\left(  \overline{\boldsymbol{a}}\right)  \right\vert
p^{-\left\vert J\right\vert -\sum_{\left(  i,j\right)  \in K(\overline
{\boldsymbol{a}})}s_{ij}}%
{\displaystyle\int\limits_{\left(  \mathbb{Z}_{p}\right)  ^{\left\vert
K_{\text{list}}(\overline{\boldsymbol{a}})\right\vert }}}
{\displaystyle\prod\limits_{\left(  i,j\right)  \in K(\overline{\boldsymbol{a}%
})}}
\left\vert x_{i}-x_{j}\right\vert _{p}^{s_{ij}}%
{\displaystyle\prod\limits_{i\in K_{\text{list}}(\overline{\boldsymbol{a}})}}
dx_{i}\\
+\left\vert \overline{\Delta}\left(  J\right)  \right\vert p^{-\left\vert
J\right\vert }.
\end{multline*}

\end{proof}

\begin{lemma}
\label{Lemma_partition}We use all the notation~introduced in Remark
\ref{Nota_basic_const}. Given $\overline{\boldsymbol{a}}=\left(  \overline
{a}_{i}\right)  _{i\in J}\in(\mathbb{F}_{p}^{\mathbb{\times}})^{\left\vert
J\right\vert }\smallsetminus\overline{\Delta}\left(  J\right)  $ and $\left(
i,j\right)  \in K(\overline{\boldsymbol{a}})$, we set%
\[
K(\left(  i,j\right)  ,\overline{\boldsymbol{a}}):=\left\{  \left(
\widetilde{i},\widetilde{j}\right)  \in K(\overline{\boldsymbol{a}}%
);\overline{a}_{i}=\overline{a}_{\widetilde{i}}\right\}
\]
and use $K_{\text{list}}(\left(  i,j\right)  ,\overline{\boldsymbol{a}%
}):=K(\left(  i,j\right)  ,\overline{\boldsymbol{a}})_{\text{list}}$. Then the
following assertions hold:

\noindent(i)
\[
K(\left(  i,j\right)  ,\overline{\boldsymbol{a}})=T_{K_{\text{list}}(\left(
i,j\right)  ,\overline{\boldsymbol{a}})}=\left\{  \left(  r,s\right)  ;2\leq
r<s\leq N-2,r,s\in K_{\text{list}}(\left(  i,j\right)  ,\overline
{\boldsymbol{a}})\right\}  ;
\]
\noindent(ii) the subsets $K(\left(  i,j\right)  ,\overline{\boldsymbol{a}})$
form a partition of $K(\overline{\boldsymbol{a}})$, i.e. there exists a finite
set $\mathcal{R}\left(  \overline{\boldsymbol{a}}\right)  $ of elements
$\left(  i,j\right)  \in K(\overline{\boldsymbol{a}}),$ such that
$K(\overline{\boldsymbol{a}})=%
{\textstyle\bigsqcup\nolimits_{\left(  i,j\right)  \in\mathcal{R}\left(
\overline{\boldsymbol{a}}\right)  }}
K(\left(  i,j\right)  ,\overline{\boldsymbol{a}}).$
\end{lemma}

\begin{proof}
(i) By definition $K(\left(  i,j\right)  ,\overline{\boldsymbol{a}})\subseteq
T_{K_{\text{list}}(\left(  i,j\right)  ,\overline{\boldsymbol{a}})}.$
Conversely, let $\left(  \widetilde{i}_{m},\widetilde{j}_{l}\right)  \in
T_{K_{\text{list}}(\left(  i,j\right)  ,\overline{\boldsymbol{a}})},$ then
there exists $\widetilde{j}_{m}\in K_{\text{list}}(\left(  i,j\right)
,\overline{\boldsymbol{a}})$ such that $\left(  \widetilde{i}_{m}%
,\widetilde{j}_{m}\right)  \in K(\left(  i,j\right)  ,\overline{\boldsymbol{a}%
})$ or $\left(  \widetilde{j}_{m},\widetilde{i}_{m}\right)  \in K(\left(
i,j\right)  ,\overline{\boldsymbol{a}})$. In any case, either $\left(
\widetilde{i}_{m},\widetilde{j}_{m}\right)  $ or $\left(  \widetilde{j}%
_{m},\widetilde{i}_{m}\right)  $ belongs to $K(\overline{\boldsymbol{a}})$ and
$\overline{a}_{i}=\overline{a}_{\widetilde{i}_{m}}=\overline{a}_{\widetilde
{j}_{m}}$. Similarly, there exists $\widetilde{i}_{l}\in K_{\text{list}%
}(\left(  i,j\right)  ,\overline{\boldsymbol{a}})$ such that either $\left(
\widetilde{i}_{l},\widetilde{j}_{l}\right)  $ or $\left(  \widetilde{j}%
_{l},\widetilde{i}_{l}\right)  $ belongs to $K(\left(  i,j\right)
,\overline{\boldsymbol{a}})$ and $\overline{a}_{i}=\overline{a}_{\widetilde
{i}_{l}}=\overline{a}_{\widetilde{j}_{l}}$. Therefore $\overline
{a}_{\widetilde{i}_{m}}=\overline{a}_{\widetilde{j}_{l}}$ i.e. $\left(
\widetilde{i}_{m},\widetilde{j}_{l}\right)  \in K(\overline{\boldsymbol{a}})$,
furthermore $\left(  \widetilde{i}_{m},\widetilde{j}_{l}\right)  \in K(\left(
i,j\right)  ,\overline{\boldsymbol{a}}).$ Hence $K(\left(  i,j\right)
,\overline{\boldsymbol{a}})=T_{K_{\text{list}}(\left(  i,j\right)
,\overline{\boldsymbol{a}})}$.

(ii) Let $\left(  i_{m},j_{m}\right)  \in K(\left(  i,j\right)  ,\overline
{\boldsymbol{a}})\cap K(\left(  \widetilde{i},\widetilde{j}\right)
,\overline{\boldsymbol{a}})$, then $\overline{a}_{i}=\overline{a}_{i_{m}%
}=\overline{a}_{\widetilde{i}}$ and $\left(  \widetilde{i},\widetilde
{j}\right)  \in$ $K(\left(  i,j\right)  ,\overline{\boldsymbol{a}})$, and
consequently $K(\left(  \widetilde{i},\widetilde{j}\right)  ,\overline
{\boldsymbol{a}})\subseteq K(\left(  i,j\right)  ,\overline{\boldsymbol{a}})$.
Similarly, one verifies that $K(\left(  i,j\right)  ,\overline{\boldsymbol{a}%
})\subseteq K(\left(  \widetilde{i},\widetilde{j}\right)  ,\overline
{\boldsymbol{a}})$.
\end{proof}

\begin{remark}
As a consequence of Lemmas \ref{Lemma_L_0}-\ref{Lemma_partition}, we have%
\[
L_{1}\left(  \boldsymbol{s};K_{\text{list}}(\overline{\boldsymbol{a}%
}),K(\overline{\boldsymbol{a}})\right)  =%
{\textstyle\prod\limits_{\left(  i,j\right)  \in\mathcal{R}\left(
\overline{\boldsymbol{a}}\right)  }}
L_{1}\left(  \boldsymbol{s};K_{\text{list}}(\left(  i,j\right)  ,\overline
{\boldsymbol{a}}),T_{K_{\text{list}}(\left(  i,j\right)  ,\overline
{\boldsymbol{a}})}\right)  .
\]

\end{remark}

\begin{example}
Take $p\geq3$, $\overline{\boldsymbol{a}}=\left(  \overline{1},\overline
{2},\overline{1},\overline{2},\overline{2}\right)  \in\mathbb{F}_{p}^{5}$, and
$J=\left\{  2,3,4,5,6\right\}  $. Hence
\[
T_{J}=\left\{  \left(  2,3\right)  ,\left(  2,4\right)  ,\left(  2,5\right)
,\left(  2,6\right)  ,\left(  3,4\right)  ,\left(  3,5\right)  ,\left(
3,6\right)  ,\left(  4,5\right)  ,\left(  4,6\right)  ,\left(  5,6\right)
\right\}  ,
\]
and by Lemma \ref{Lemma_partition},
\[
K\left(  \overline{\boldsymbol{a}}\right)  =\left\{  \left(  2,4\right)
,\left(  3,5\right)  ,\left(  3,6\right)  ,\left(  5,6\right)  \right\}
=K(\left(  2,4\right)  ,\overline{\boldsymbol{a}})%
{\textstyle\bigsqcup}
K(\left(  3,5\right)  ,\overline{\boldsymbol{a}}),
\]
where $K(\left(  2,4\right)  ,\overline{\boldsymbol{a}})=\left\{  \left(
2,4\right)  \right\}  ,$ $K(\left(  3,5\right)  ,\overline{\boldsymbol{a}%
})=\left\{  \left(  3,5\right)  ,\left(  3,6\right)  ,\left(  5,6\right)
\right\}  $. Thus
\[
K_{list}(\left(  2,4\right)  ,\overline{\boldsymbol{a}})=\left\{  2,4\right\}
\text{ and }K_{\text{list}}(\left(  3,5\right)  ,\overline{\boldsymbol{a}%
})\allowbreak=\left\{  3,5,6\right\}  .
\]
With this notation, $L_{1}\left(  \boldsymbol{s};K_{\text{list}}%
(\overline{\boldsymbol{a}}),K(\overline{\boldsymbol{a}})\right)  $ equals
\begin{gather*}%
{\displaystyle\int\limits_{\mathbb{Z}_{p}^{5}}}
\left\vert x_{2}-x_{4}\right\vert _{p}^{s_{24}}\left\vert x_{3}-x_{5}%
\right\vert _{p}^{s_{35}}\left\vert x_{3}-x_{6}\right\vert _{p}^{s_{36}%
}\left\vert x_{5}-x_{6}\right\vert _{p}^{s_{56}}dx_{2}dx_{3}dx_{4}dx_{5}%
dx_{6}\\
=\left\{
{\displaystyle\int\limits_{\mathbb{Z}_{p}^{2}}}
\left\vert x_{2}-x_{4}\right\vert _{p}^{s_{24}}dx_{2}dx_{4}\right\}  \left\{
{\displaystyle\int\limits_{\mathbb{Z}_{p}^{3}}}
\left\vert x_{3}-x_{5}\right\vert _{p}^{s_{35}}\left\vert x_{3}-x_{6}%
\right\vert _{p}^{s_{36}}\left\vert x_{5}-x_{6}\right\vert _{p}^{s_{56}}%
dx_{3}dx_{5}dx_{6}\right\} \\
=L_{1}\left(  \boldsymbol{s};K_{\text{list}}(\left(  2,4\right)
,\overline{\boldsymbol{a}}),T_{K_{\text{list}}(\left(  2,4\right)
,\overline{\boldsymbol{a}})}\right)  L_{1}\left(  \boldsymbol{s}%
;K_{\text{list}}(\left(  3,5\right)  ,\overline{\boldsymbol{a}}%
),T_{K_{\text{list}}(\left(  3,5\right)  ,\overline{\boldsymbol{a}})}\right)
.
\end{gather*}

\end{example}

\begin{lemma}
\label{Lemma_F}Set $F(s_{1},s_{2},s_{3},x,y):=|x|_{p}^{s_{1}}|y|_{p}^{s_{2}%
}|x-y|_{p}^{s_{3}}$, $s_{1}$, $s_{2}$, $s_{3}\in\mathbb{C}$, and%
\[
Z\left(  s_{1},s_{2},s_{3}\right)  :=\int_{%
\mathbb{Z}
_{p}^{2}}F(s_{1},s_{2},s_{3},x,y)dxdy\text{ for }\operatorname{Re}%
(s_{i})>0\text{, }i=1,2,3\text{.}%
\]
Then $Z\left(  s_{1},s_{2},s_{3}\right)  $ is a holomorphic function on
\[
\left\{  \left(  s_{1},s_{2},s_{3}\right)  \in\mathbb{C}^{3};\operatorname{Re}%
(s_{i})>-1\text{ for }i=1,2,3\text{ and }\operatorname{Re}(s_{1}%
)+\operatorname{Re}(s_{2})+\operatorname{Re}(s_{3})>-2\right\}  .
\]
In addition,%
\[
Z\left(  s_{1},s_{2},s_{3}\right)  :=\frac{Q\left(  p^{-s_{1}},p^{-s_{2}%
},p^{-s_{3}}\right)  }{\left(  1-p^{-2-s_{1}-s_{2}-s_{3}}\right)
\prod\limits_{i=1}^{3}\left(  1-p^{-1-s_{i}}\right)  },
\]
where $Q\left(  p^{-s_{1}},p^{-s_{2}},p^{-s_{3}}\right)  $ denotes a
polynomial with rational coefficients in the variables $p^{-s_{1}}$,
$p^{-s_{2}}$, $p^{-s_{3}}$.
\end{lemma}

\begin{remark}
If $s_{1}=s_{2}=0$, then the denominator of $Z\left(  s_{1},s_{2}%
,s_{3}\right)  $ is $1-p^{-1-s_{3}}$.
\end{remark}

\begin{proof}
By using that $\mathbb{Z}_{p}^{2}=(p\mathbb{Z}_{p})^{2}\sqcup S_{0}^{2}$ with
$S_{0}^{2}=p\mathbb{Z}_{p}\times\mathbb{Z}_{p}^{\times}\sqcup\mathbb{Z}%
_{p}^{\times}\times p\mathbb{Z}_{p}\sqcup\mathbb{Z}_{p}^{\times}%
\times\mathbb{Z}_{p}^{\times}$, and then by changing variables, we get%
\[
Z\left(  s_{1},s_{2},s_{3}\right)  =\frac{\int_{S_{0}^{2}}F(s_{1},s_{2}%
,s_{3},x,y)dxdy}{1-p^{-2-s_{1}-s_{2}-s_{3}}}=:\frac{Z_{0}\left(  s_{1}%
,s_{2},s_{3}\right)  }{1-p^{-2-s_{1}-s_{2}-s_{3}}}.
\]
On the other hand,
\begin{gather*}
Z_{0}\left(  s_{1},s_{2},s_{3}\right)  =\int_{p\mathbb{Z}_{p}\times
\mathbb{Z}_{p}^{\times}}F(s_{1},s_{2},s_{3},x,y)dxdy\\
+\int_{\mathbb{Z}_{p}^{\times}\times p\mathbb{Z}_{p}}F(s_{1},s_{2}%
,s_{3},x,y)dxdy+\int_{\mathbb{Z}_{p}^{\times}\times\mathbb{Z}_{p}^{\times}%
}F(s_{1},s_{2},s_{3},x,y)dxdy\\
=:Z_{0,1}(s_{1},s_{2},s_{3})+Z_{0,2}(s_{1},s_{2},s_{3})+Z_{0},_{3}(s_{1}%
,s_{2},s_{3}).
\end{gather*}
First, we compute $Z_{0,1}(s_{1},s_{2},s_{3})$. By a change of variables, we
get
\[
Z_{0,1}(s_{1},s_{2},s_{3})=p^{-1-s_{1}}(1-p^{-1})\int_{\mathbb{Z}_{p}}%
|x|_{p}^{s_{1}}dx=\frac{\left(  1-p^{-1}\right)  ^{2}p^{-1-s_{1}}%
}{1-p^{-1-s_{1}}}%
\]
for $\operatorname{Re}(s_{1})>-1$. By a similar computation we obtain
\[
Z_{0,2}(s_{1},s_{2},s_{3})=\frac{\left(  1-p^{-1}\right)  ^{2}p^{-1-s_{2}}%
}{1-p^{-1-s_{2}}}\text{ for }\operatorname{Re}(s_{2})>-1.
\]
In order to compute
\[
Z_{0,3}(s_{1},s_{2},s_{3})=\int_{\mathbb{Z}_{p}^{\times}\times\mathbb{Z}%
_{p}^{\times}}|x-y|_{p}^{s_{3}}dxdy,
\]
we use that $(\mathbb{Z}_{p}^{\times})^{2}=\sqcup_{\overline{a}_{0}%
,\overline{a}_{1}\in\mathbb{F}_{p}^{\times}}a_{0}+p\mathbb{Z}_{p}\times
a_{1}+p\mathbb{Z}_{p}$, where $\mathbb{F}_{p}^{\times}=\left\{
1,2,...,p-1\right\}  $ as sets, to get%
\begin{gather*}
Z_{0,3}(s_{1},s_{2},s_{3})=\sum_{\overline{a}_{0},\overline{a}_{1}%
\in\mathbb{F}_{p}^{\times}}\int_{a_{0}+p\mathbb{Z}_{p}\times a_{1}%
+p\mathbb{Z}_{p}}|x-y|_{p}^{s_{3}}dxdy\\
=p^{-2}\sum_{\substack{\overline{a}_{0},\overline{a}_{1}\in\mathbb{F}%
_{p}^{\times} \\\overline{a}_{0}\neq\overline{a}_{1}}}\int_{\mathbb{Z}%
_{p}\times\mathbb{Z}_{p}}|a_{0}+px-a_{1}-py|_{p}^{s_{3}}dxdy+p^{-2}%
\sum_{\substack{\overline{a}_{0},\overline{a}_{1}\in\mathbb{F}_{p}^{\times}
\\\overline{a}_{0}=\overline{a}_{1}}}\int_{\mathbb{Z}_{p}\times\mathbb{Z}_{p}%
}|x-y|_{p}^{s_{3}}dxdy\\
=p^{-2}(p-1)(p-2)+p^{-2-s_{3}}\left(  p-1\right)  \frac{1-p^{-1}%
}{1-p^{-1-s_{3}}}.
\end{gather*}

\end{proof}

\begin{lemma}
\label{Lemma_L_1}Let $I$ be a subset of $T$ satisfying $\left\vert
I\right\vert \geq2$. Then $L_{1}\left(  \boldsymbol{s};I,T_{I}\right)  $
admits an analytic continuation as a rational function of the form%
\begin{equation}
L_{1}\left(  \boldsymbol{s};I,T_{I}\right)  =\frac{Q_{I}\left(  \left\{
p^{-s_{ij}}\right\}  _{i,j\in I}\right)  }{\prod\limits_{J\in\mathcal{F}%
\left(  I\right)  }\left(  1-p^{-\left(  \left\vert J\right\vert
-1+\sum_{\substack{2\leq i<j\leq N-2 \\i,j\in J}}s_{ij}\right)  }\right)
^{e_{J}}\prod\limits_{ij\in S_{I}}\left(  1-p^{-1-s_{ij}}\right)  ^{e_{ij}}},
\label{Formula_L_N_1}%
\end{equation}
where $Q_{I}\left(  \left\{  p^{-s_{ij}}\right\}  _{i,j\in I}\right)  $ is a
polynomial with rational coefficients in the variables $\left\{  p^{-s_{ij}%
}\right\}  _{i,j\in I}$, $\mathcal{F}(I)$ is a family of subsets of $I$, with
$I\in\mathcal{F}(I)$, $S_{I}$ is a non-empty subset of $\left\{  2\leq i<j\leq
N-2,i,j\in I\right\}  $, and the $e_{J},$ $e_{ij}$'s are positive integers.
\end{lemma}

\begin{proof}
By using the partition $\mathbb{Z}_{p}^{\left\vert I\right\vert }%
=(p\mathbb{Z}_{p})^{\left\vert I\right\vert }\sqcup S_{0}^{\left\vert
I\right\vert }$, where $\mathbb{Z}_{p}^{\left\vert I\right\vert }=\left\{
\left(  x_{i}\right)  _{i\in I};x_{i}\in\mathbb{Z}_{p}\right\}  $,
$(p\mathbb{Z}_{p}\mathbb{)}^{\left\vert I\right\vert }=\left\{  \left(
x_{i}\right)  _{i\in I};x_{i}\in p\mathbb{Z}_{p}\right\}  $, and
$S_{0}^{\left\vert I\right\vert }=\left\{  \left(  x_{i}\right)  _{i\in I}%
\in\mathbb{Z}_{p}^{\left\vert I\right\vert };\max_{i\in I}\left\{  \left\vert
x_{i}\right\vert _{p}\right\}  =1\right\}  $. By a change of variables, we get%
\begin{align*}
L_{1}\left(  \boldsymbol{s};I,T_{I}\right)   &  =\frac{%
{\displaystyle\int\limits_{S_{0}^{\left\vert I\right\vert }}}
{\displaystyle\prod\limits_{\substack{2\leq i<j\leq N-2 \\i,j\in I}}}
\left\vert x_{i}-x_{j}\right\vert _{p}^{s_{ij}}%
{\displaystyle\prod\limits_{i\in J}}
dx_{i}}{1-p^{-\left\vert I\right\vert -\sum_{\substack{2\leq i<j\leq N-2
\\i,j\in I}}s_{ij}}}\\
&  =:\frac{B_{0}\left(  \boldsymbol{s};I\right)  }{1-p^{-\left\vert
I\right\vert -\sum_{\substack{2\leq i<j\leq N-2 \\i,j\in I}}s_{ij}}}.
\end{align*}
For every non-empty subset $J\subseteq I$, we define%
\[
S_{J}^{\left\vert I\right\vert }=\left\{  \left(  x_{i}\right)  _{i\in I}%
\in\mathbb{Z}_{p}^{\left\vert I\right\vert };\left\vert x_{i}\right\vert
_{p}=1\Leftrightarrow i\in J\right\}  ,
\]
then $S_{0}^{\left\vert I\right\vert }=\sqcup_{J\subseteq I,J\neq\varnothing
}S_{J}^{\left\vert I\right\vert }$ and $A_{0}\left(  \boldsymbol{s};I\right)
=\sum_{_{J\subseteq I,J\neq\varnothing}}A_{0,J}\left(  \boldsymbol{s}\right)
$ where
\[
B_{0,J}\left(  \boldsymbol{s}\right)  :=%
{\displaystyle\int\limits_{S_{J}^{\left\vert I\right\vert }}}
{\displaystyle\prod\limits_{\substack{2\leq i<j\leq N-2 \\i,j\in I}}}
\left\vert x_{i}-x_{j}\right\vert _{p}^{s_{ij}}%
{\displaystyle\prod\limits_{i\in I}}
dx_{i},
\]
for this reason
\begin{equation}
L_{1}\left(  \boldsymbol{s};I,T_{I}\right)  =\frac{B_{0,I}\left(
\boldsymbol{s}\right)  +\sum_{_{J\subsetneqq I,J\neq\varnothing}}%
B_{0,J}\left(  \boldsymbol{s}\right)  }{1-p^{-\left\vert I\right\vert
-\sum_{\substack{2\leq i<j\leq N-2 \\i,j\in I}}s_{ij}}}. \label{Form_L_1}%
\end{equation}
On the other hand,%
\begin{equation}
\left\vert x_{i}-x_{j}\right\vert _{p}^{s_{ij}}\mid_{S_{J}^{\left\vert
I\right\vert }}=\left\{
\begin{array}
[c]{lll}%
\left\vert x_{_{i}}-x_{j}\right\vert _{p}^{s_{ij}} & \text{if} & i,j\in J\\
&  & \\
\left\vert x_{_{i}}-x_{j}\right\vert _{p}^{s_{ij}} & \text{if} & i,j\in
I\smallsetminus J\\
&  & \\
1 & \text{if} & i\in J,j\in I\smallsetminus J\\
&  & \\
1 & \text{if} & j\in J,i\in I\smallsetminus J.
\end{array}
\right.  \label{Calcul_abs}%
\end{equation}
Then
\[
B_{0,I}\left(  \boldsymbol{s}\right)  =L_{0}\left(  \boldsymbol{s};I\right)
,
\]
and if $J\subsetneqq I$,%
\begin{align}
B_{0,J}\left(  \boldsymbol{s}\right)   &  =\left\{
{\displaystyle\int\limits_{(p\mathbb{Z}_{p})^{\left\vert I\smallsetminus
J\right\vert }}}
{\displaystyle\prod\limits_{\substack{2\leq i<j\leq N-2 \\i,j\in
I\smallsetminus J}}}
\left\vert x_{i}-x_{j}\right\vert _{p}^{s_{ij}}%
{\displaystyle\prod\limits_{i\in I\smallsetminus J}}
dx_{i}\right\}  L_{0}\left(  \boldsymbol{s};J\right) \label{Form_A_0_J}\\
&  =p^{-\left\vert I\smallsetminus J\right\vert -\sum_{_{\substack{2\leq
i<j\leq N-2 \\i,j\in I\smallsetminus J}}}s_{ij}}L_{1}\left(  \boldsymbol{s}%
;I\smallsetminus J,T_{I\smallsetminus J}\right)  L_{0}\left(  \boldsymbol{s}%
;J\right)  .\nonumber
\end{align}
Therefore, from (\ref{Form_L_1})-(\ref{Form_A_0_J}), $L_{1}\left(
\boldsymbol{s};I,T_{I}\right)  $ equals
\begin{equation}
\frac{L_{0}\left(  \boldsymbol{s};I\right)  +\sum_{_{J\subsetneqq
I,J\neq\varnothing}}p^{-\left\vert I\smallsetminus J\right\vert -\sum
_{_{\substack{2\leq i<j\leq N-2 \\i,j\in I\smallsetminus J}}}s_{ij}}%
L_{1}\left(  \boldsymbol{s};I\smallsetminus J,T_{I\smallsetminus J}\right)
L_{0}\left(  \boldsymbol{s};J\right)  }{1-p^{-\left\vert I\right\vert
-\sum_{\substack{2\leq i<j\leq N-2 \\i,j\in I}}s_{ij}}}. \label{Formula_1}%
\end{equation}
Now, by Lemma \ref{Lemma_L_0} and \ the fact that $\overline{A}\left(
\overline{\boldsymbol{1}}\right)  =%
{\textstyle\bigsqcup\nolimits_{\overline{b}\in\mathbb{F}_{p}^{\mathbb{\times}%
}}}
\left\{  \left(  \overline{b}\right)  _{i\in I}\right\}  $, $K_{\text{list}%
}(\overline{\boldsymbol{1}})=I$, see Remark \ref{Nota_basic_const},%
\begin{gather}
L_{0}\left(  \boldsymbol{s};I\right)  =\sum\limits_{\overline{\boldsymbol{a}%
}\in\mathcal{R}(I)\smallsetminus\left\{  \overline{\boldsymbol{1}}\right\}
}\left\vert \overline{A}(\overline{\boldsymbol{a}})\right\vert p^{-\left\vert
I\right\vert -\sum_{_{\left(  i,j\right)  \in K(\overline{\boldsymbol{a}})}%
}s_{ij}}L_{1}\left(  \boldsymbol{s};K_{\text{list}}(\overline{\boldsymbol{a}%
}),K(\overline{\boldsymbol{a}})\right) \label{Formula_2}\\
+\left(  p-1\right)  p^{-\left\vert I\right\vert -\sum_{\substack{2\leq
i<j\leq N-2 \\i,j\in I}}s_{ij}}L_{1}\left(  \boldsymbol{s};I,T_{I}\right)
+\left\vert \overline{\Delta}\left(  I\right)  \right\vert p^{-\left\vert
I\right\vert },\nonumber
\end{gather}
with $\left\vert K_{\text{list}}(\overline{\boldsymbol{a}})\right\vert \geq2
$, hence from (\ref{Formula_1})-(\ref{Formula_2}),
\begin{gather*}
\left(  1-p^{1-\left\vert I\right\vert -\sum_{\substack{2\leq i<j\leq N-2
\\i,j\in I}}s_{ij}}\right)  L_{1}\left(  \boldsymbol{s};I,T_{I}\right) \\
=\sum\limits_{\overline{\boldsymbol{a}}\in\mathcal{R}(I)\smallsetminus\left\{
\overline{\boldsymbol{1}}\right\}  }d_{\overline{\boldsymbol{a}}}\left(
\boldsymbol{s}\right)  L_{1}\left(  \boldsymbol{s};K_{\text{list}}%
(\overline{\boldsymbol{a}}),K(\overline{\boldsymbol{a}})\right) \\
+\sum_{_{_{\substack{J\subsetneqq I \\J\neq\varnothing}}}}c_{J}(\boldsymbol{s}%
)L_{1}\left(  \boldsymbol{s};I\smallsetminus J,T_{I\smallsetminus J}\right)
L_{0}\left(  \boldsymbol{s};J\right)  +\left\vert \overline{\Delta}\left(
I\right)  \right\vert p^{-\left\vert I\right\vert }.
\end{gather*}
This formula and Lemmas \ref{Lemma_L_0}-\ref{Lemma_F} give a recursive
algorithm for computing integrals $L_{1}\left(  \boldsymbol{s};I,T_{I}\right)
$, from which we get (\ref{Formula_L_N_1}).
\end{proof}

From Lemmas \ref{Lemma_L_0}-\ref{Lemma_L_1}, we obtain the following result:

\begin{corollary}
\label{Cor_Lemma_L_1}If $\left\vert I\right\vert \geq2$, then%
\[
L_{0}\left(  \boldsymbol{s};I\right)  =\frac{R_{I}\left(  \left\{  p^{-s_{ij}%
}\right\}  _{i,j\in I}\right)  }{\prod\limits_{J\in\mathcal{G}\left(
I\right)  }\left(  1-p^{-\left(  \left\vert J\right\vert -1+\sum
_{\substack{2\leq i<j\leq N-2 \\i,j\in J}}s_{ij}\right)  }\right)  ^{f_{J}%
}\prod\limits_{ij\in G_{I}}\left(  1-p^{-1-s_{ij}}\right)  ^{f_{ij}}},
\]
where $R_{I}\left(  \left\{  p^{-s_{ij}}\right\}  _{i,j\in I}\right)  $ is a
polynomial with rational coefficients in the variables $\left\{  p^{-s_{ij}%
}\right\}  _{i,j\in I}$, $\mathcal{G}\left(  I\right)  $ is a family of
non-empty subsets of $I$, with $I\in\mathcal{G}\left(  I\right)  $, $G_{I}$ is
a non-empty subset of $\left\{  2\leq i<j\leq N-2,i,j\in I\right\}  ,$ and the
$f_{J}$, $f_{ij}$ 's are positive integers.
\end{corollary}

Given $I\subseteq T$, with $\left\vert I\right\vert \geq2$, and $K\subseteq
I$, with \ $\left\vert K\right\vert \geq1$, and $M\subseteq T_{I}$, with
$\left\vert M\right\vert \geq1$, we define%
\[
L_{2}\left(  \boldsymbol{s};I,K,M\right)  =%
{\displaystyle\int\limits_{\mathbb{Z}_{p}^{\left\vert I\right\vert }}}
{\displaystyle\prod\limits_{i\in K}}
\left\vert x_{i}\right\vert _{p}^{s_{ti}}%
{\displaystyle\prod\limits_{\left(  i,j\right)  \in M}}
\left\vert x_{i}-x_{j}\right\vert _{p}^{s_{ij}}%
{\displaystyle\prod\limits_{i\in I}}
dx_{i}\text{ }%
\]
for $\operatorname{Re}\left(  s_{ij}\right)  >0$ for any $ij$. If $\left\vert
M\right\vert =0$, then
\[
L_{2}\left(  \boldsymbol{s};I,K,M\right)  =\int_{\mathbb{Z}_{p}^{\left\vert
I\right\vert }}%
{\textstyle\prod\nolimits_{i\in K}}
\left\vert x_{i}\right\vert _{p}^{s_{ti}}%
{\textstyle\prod\nolimits_{i\in I}}
dx_{i}.
\]

\begin{lemma}
\label{Lema_L_2}Let $t\in\left\{  1,N-1\right\}  $. Then
\[
L_{2}\left(  \boldsymbol{s};I,K,T_{I}\right)  =%
{\displaystyle\int\limits_{\mathbb{Z}_{p}^{\left\vert I\right\vert }}}
{\displaystyle\prod\limits_{i\in K}}
\left\vert x_{i}\right\vert _{p}^{s_{ti}}%
{\displaystyle\prod\limits_{\substack{2\leq i<j\leq N-2 \\i,j\in I}}}
\left\vert x_{i}-x_{j}\right\vert _{p}^{s_{ij}}%
{\displaystyle\prod\limits_{i\in I}}
dx_{i}\text{ }%
\]
admits an analytic continuation as a rational function of the form%
\begin{equation}
L_{2}\left(  \boldsymbol{s};I,K,T_{I}\right)  =\frac{Q_{I,K}\left(  \left\{
p^{-s_{ij}}\right\}  _{i,j\in I},\left\{  p^{-s_{ti}}\right\}  _{t\in\left\{
1,N-1\right\}  ,i\in I}\right)  }{R_{0}(\boldsymbol{s};I,K)R_{1}%
(\boldsymbol{s};I,K)R_{2}(\boldsymbol{s};I,K)}, \label{Formula_L_2_N}%
\end{equation}
where%
\[
R_{0}(\boldsymbol{s};I,K)=\prod\limits_{J\in\mathcal{G}_{1}\left(  I\right)
}\left(  1-p^{-\left(  \left\vert J\right\vert -1+\sum_{\substack{2\leq
i<j\leq N-2 \\i,j\in J}}s_{ij}\right)  }\right)  ^{f_{J}}\prod\limits_{ij\in
S_{I}}\left(  1-p^{-1-s_{ij}}\right)  ^{g_{ij}},
\]%
\[
R_{1}(\boldsymbol{s};I,K)=\prod\limits_{i\in U_{K}}\left(  1-p^{-1-s_{ti}%
}\right)  ^{h_{i}},
\]%
\[
R_{2}(\boldsymbol{s};I,K)=\prod\limits_{\left(  J,R\right)  \in\mathcal{G}%
_{2}\left(  I\times I\right)  }\left(  1-p^{-\left\vert J\right\vert
-\sum_{i\in R}s_{ti}-\sum\nolimits_{2\leq i<j\leq N-2,\text{ }i,j\in J}s_{ij}%
}\right)  ,
\]
where $Q_{I,K}\left(  \left\{  p^{-s_{ij}}\right\}  _{i,j\in I},\left\{
p^{-s_{ti}}\right\}  _{t\in\left\{  1,N-1\right\}  ,i\in I}\right)  $ denotes
a polynomial with rational coefficients in the variables $\left\{  p^{-s_{ij}%
}\right\}  _{i,j\in I},\left\{  p^{-s_{ti}}\right\}  _{t\in\left\{
1,N-1\right\}  ,i\in I}$, $\mathcal{G}_{1}\left(  I\right)  $ is a non-empty
family of subsets of $I$, with $I\in\mathcal{G}_{1}\left(  I\right)  $,
$\mathcal{G}_{2}\left(  I\times I\right)  $ is a non-empty family of subsets
$J\times R$ of $I\times I$, with $R\subseteq J$ and $\left(  I,K\right)
\in\mathcal{G}_{2}\left(  I\times I\right)  $, $U_{K}$ is a non-empty subset
of $K $, $S_{I}$ is a non-empty subset of $\left\{  2\leq i<j\leq N-2,i,j\in
I\right\}  $, and the $f_{J}$'s, $g_{ij}$'s, and the $h_{i}$'s are positive integers.
\end{lemma}

\begin{remark}
\label{Nota_Lemma_L_2}The integral $L_{2}\left(  \boldsymbol{s};I,K,M\right)
$ is also a multivariate $p$-adic local zeta function. If $\left\vert
I\right\vert \geq2$ and $\left\vert K\right\vert =0$, then $L_{2}\left(
\boldsymbol{s};I,K,M\right)  =L_{1}\left(  \boldsymbol{s};I,M\right)  $.
\end{remark}

\begin{proof}
We use the partition $\mathbb{Z}_{p}^{\left\vert I\right\vert }=(p\mathbb{Z}%
_{p})^{\left\vert I\right\vert }\sqcup S_{0}^{\left\vert I\right\vert }$ as in
the proof of Lemma \ref{Lemma_L_1} and a change of variables, to get%
\begin{align*}
L_{2}\left(  \boldsymbol{s};I,K,T_{I}\right)   &  =\frac{%
{\displaystyle\int\limits_{S_{0}^{\left\vert I\right\vert }}}
{\displaystyle\prod\limits_{i\in K}}
\left\vert x_{i}\right\vert _{p}^{s_{ti}}%
{\displaystyle\prod\limits_{2\leq i<j\leq N-2,i,j\in I}}
\left\vert x_{i}-x_{j}\right\vert _{p}^{s_{ij}}%
{\displaystyle\prod\limits_{i\in I}}
dx_{i}}{1-p^{-\left\vert I\right\vert -\sum_{i\in K}s_{ti}-\sum
\nolimits_{2\leq i<j\leq N-2,i,j\in I}s_{ij}}}\\
&  =:\frac{B_{0}\left(  \boldsymbol{s};I,K,T_{I}\right)  }{1-p^{-\left\vert
I\right\vert -\sum_{i\in K}s_{ti}-\sum\nolimits_{2\leq i<j\leq N-2,\text{
}i,j\in I}s_{ij}}}.
\end{align*}
We now use the partition $S_{0}^{\left\vert I\right\vert }=\sqcup_{J\subseteq
I,J\neq\varnothing}S_{J}^{\left\vert I\right\vert }$ to obtain
\[
B_{0}\left(  \boldsymbol{s};I,K,T_{I}\right)  =\sum_{_{J\subseteq
I,J\neq\varnothing}}B_{0,J}\left(  \boldsymbol{s}\right)  ,
\]
where
\[
B_{0,J}\left(  \boldsymbol{s}\right)  :=%
{\displaystyle\int\limits_{S_{J}^{\left\vert I\right\vert }}}
{\displaystyle\prod\limits_{i\in K}}
\left\vert x_{i}\right\vert _{p}^{s_{ti}}%
{\displaystyle\prod\limits_{\substack{2\leq i<j\leq N-2 \\i,j\in I}}}
\left\vert x_{i}-x_{j}\right\vert _{p}^{s_{ij}}%
{\displaystyle\prod\limits_{i\in I}}
dx_{i}.
\]
Consequently
\[
L_{2}\left(  \boldsymbol{s};I,K,T_{I}\right)  =\frac{B_{0,I}\left(
\boldsymbol{s}\right)  +\sum_{_{J\subsetneqq I,J\neq\varnothing}}%
B_{0,J}\left(  \boldsymbol{s}\right)  }{1-p^{-\left\vert I\right\vert
-\sum_{i\in K}s_{ti}-\sum\nolimits_{2\leq i<j\leq N-2,\text{ }i,j\in I}s_{ij}%
}}.
\]
On the other hand, $\left\vert x_{i}-x_{j}\right\vert _{p}^{s_{ij}}\mid
_{S_{J}^{\left\vert I\right\vert }}$ is given in (\ref{Calcul_abs}) and
\[%
{\textstyle\prod\nolimits_{i\in K}}
\left\vert x_{i}\right\vert _{p}^{s_{ti}}\mid_{S_{J}^{\left\vert I\right\vert
}}=%
{\textstyle\prod\nolimits_{i\in K}}
\left\vert x_{i}\right\vert _{p}^{s_{ti}}\mid_{\left(  p\mathbb{Z}_{p}\right)
^{\left\vert K\smallsetminus J\right\vert }}.
\]
Then $B_{0,I}\left(  \boldsymbol{s}\right)  =L_{0}\left(  \boldsymbol{s}%
;I\right)  $, and if $J\subsetneqq I$, $B_{0,J}\left(  \boldsymbol{s}\right)
$ equals
\[%
\begin{array}
[c]{c}%
\left\{
{\displaystyle\int\limits_{(p\mathbb{Z}_{p})^{\left\vert I\smallsetminus
J\right\vert }}}
{\displaystyle\prod\limits_{i\in K\smallsetminus J}}
\left\vert x_{i}\right\vert _{p}^{s_{ti}}%
{\displaystyle\prod\limits_{\substack{2\leq i<j\leq N-2 \\i,j\in
I\smallsetminus J}}}
\left\vert x_{i}-x_{j}\right\vert _{p}^{s_{ij}}%
{\displaystyle\prod\limits_{i\in I\smallsetminus J}}
dx_{i}\right\}  L_{0}\left(  \boldsymbol{s};J\right)  =\\
\\
p^{-\left\vert I\smallsetminus J\right\vert -\sum_{i\in K\smallsetminus
J}s_{ti}-\sum_{2\leq i<j\leq N-2,i,j\in I\smallsetminus J}s_{ij}}\times\\
\\
\left\{
{\displaystyle\int\limits_{\mathbb{Z}_{p}^{\left\vert I\smallsetminus
J\right\vert }}}
{\displaystyle\prod\limits_{i\in K\smallsetminus J}}
\left\vert x_{i}\right\vert _{p}^{s_{ti}}%
{\displaystyle\prod\limits_{\substack{2\leq i<j\leq N-2 \\i,j\in
I\smallsetminus J}}}
\left\vert x_{i}-x_{j}\right\vert _{p}^{s_{ij}}%
{\displaystyle\prod\limits_{i\in I\smallsetminus J}}
dx_{i}\right\}  L_{0}\left(  \boldsymbol{s};J\right)  =\\
\\
p^{-\left\vert I\smallsetminus J\right\vert -\sum_{i\in K\smallsetminus
J}s_{t}{}_{i}-\sum_{2\leq i<j\leq N-2,i,j\in I\smallsetminus J}s_{i}{}_{j}%
}L_{2}\left(  \boldsymbol{s};I\smallsetminus J,K\smallsetminus
J,T_{I\smallsetminus J}\right)  L_{0}\left(  \boldsymbol{s};J\right)  .
\end{array}
\]
Hence $\left(  1-p^{-\left\vert I\right\vert -\sum_{i\in K}s_{ti}%
-\sum\nolimits_{2\leq i<j\leq N-2,\text{ }i,j\in I}s_{ij}}\right)
L_{2}\left(  \boldsymbol{s};I,K,T_{I}\right)  $ equals%
\begin{gather}
L_{0}\left(  \boldsymbol{s};I\right)  +\label{Int_1}\\
\sum_{_{J\subsetneqq I,\text{ }J\neq\varnothing}}p^{-\left\vert
I\smallsetminus J\right\vert -\sum_{i\in K\smallsetminus J}s_{ti}-\sum_{2\leq
i<j\leq N-2,i,j\in I\smallsetminus J}s_{ij}}L_{2}\left(  \boldsymbol{s}%
;I\smallsetminus J,K\smallsetminus J,T_{I\smallsetminus J}\right)
L_{0}\left(  \boldsymbol{s};J\right)  .\nonumber
\end{gather}
By using that $\left\vert I\smallsetminus J\right\vert <\left\vert
I\right\vert $ if $J\subsetneqq I$, $J\neq\varnothing$, and that integrals
$L_{0}\left(  \boldsymbol{s};I\right)  $, $L_{0}\left(  \boldsymbol{s}%
;J\right)  $ can be computed effectively, see Corollary \ref{Cor_Lemma_L_1},
formula (\ref{Int_1}) gives a recursive algorithm for computing $L_{2}\left(
\boldsymbol{s};I,K,T_{I}\right)  $, by using it, we obtain
(\ref{Formula_L_2_N}). Notice the integrals of type $L_{2}\left(
\boldsymbol{s};I,K,T_{I}\right)  $, with $\left\vert I\right\vert =1$ and
$K=\left\{  i\right\}  $ contribute with terms of the form $\frac{1-p^{-1}%
}{1-p^{-1-s_{ti}}}$.
\end{proof}

\begin{lemma}
\label{Lemma_M_a_J}Given $J$ a non-empty subset of $T$, with $\left\vert
J\right\vert \geq2$, we define
\[
\boldsymbol{M}_{1}(\boldsymbol{s};J)=%
{\displaystyle\int\limits_{(\mathbb{Z}_{p}^{\times})^{\left\vert J\right\vert
}}}
{\displaystyle\prod\limits_{i\in J}}
\left\vert 1-x_{i}\right\vert _{p}^{s_{(N-1)i}}%
{\displaystyle\prod\limits_{\substack{2\leq i<j\leq N-2 \\i,j\in J}}}
\left\vert x_{i}-x_{j}\right\vert _{p}^{s_{ij}}%
{\displaystyle\prod\limits_{i\in J}}
dx_{i}%
\]
for $\operatorname{Re}\left(  s_{(N-1)i}\right)  >0$, $i\in J$, and
$\operatorname{Re}(s_{ij})>0$, for $i,j\in J$. Then, $\boldsymbol{M}%
_{1}(\boldsymbol{s};J)$\ admits an analytic continuation as a rational
function of the form%
\begin{equation}
\boldsymbol{M}_{1}(\boldsymbol{s};J)=\frac{Q_{J}\left(  \left\{  p^{-s_{ij}%
}\right\}  _{i,j\in J},\left\{  p^{-s_{(N-1)i}}\right\}  _{i\in J}\right)
}{\prod\limits_{i=1}^{3}U_{i}(\boldsymbol{s};J)}, \label{Formula_M_a_J}%
\end{equation}
where%
\[
U_{1}(\boldsymbol{s};J)=\prod\limits_{M\in\mathcal{F}_{1}\left(  J\right)
}\left(  1-p^{-\left(  \left\vert M\right\vert -1+\sum_{\substack{2\leq
i<j\leq N-2 \\i,j\in M}}s_{ij}\right)  }\right)  ^{e_{M}}\prod\limits_{ij\in
S_{J}^{\left(  1\right)  }}\left(  1-p^{-1-s_{ij}}\right)  ^{f_{ij}},
\]%
\[
U_{2}(\boldsymbol{s};J)=\prod\limits_{\left(  M,S\right)  \in\mathcal{F}%
_{2}\left(  J\right)  }\left(  1-p^{-\left\vert M\right\vert -\sum_{i\in
S}s_{(N-1)i}-\sum\nolimits_{2\leq i<j\leq N-2,\text{ }i,j\in M}s_{ij}}\right)
^{g_{\left(  M,S\right)  }},
\]
and
\[
U_{3}(\boldsymbol{s};J)=\prod\limits_{i\in S_{J}^{\left(  2\right)  }}\left(
1-p^{-1-s_{(N-1)i}}\right)  ^{h_{i}},
\]
where $\mathcal{F}_{1}\left(  J\right)  $ is a non-empty family of subsets of
$J$, with $J\in\mathcal{F}_{1}\left(  J\right)  ,$ $\mathcal{F}_{2}\left(
J\right)  $ is a non-empty family of subsets $M\times S\subseteq$ $J\times J$,
with $S\subseteq M$, $S_{J}^{\left(  1\right)  }$and $S_{J}^{\left(  2\right)
}$ are non-empty subsets of $T$, and the $e_{M}$'s, $f_{ij}$'s, $g_{\left(
M,S\right)  }$'s and the $h_{i}$'s are positive integers.
\end{lemma}

\begin{remark}
\label{Nota_Lemma_M_a_J}If $\left\vert J\right\vert =1$, then $\boldsymbol{M}%
_{1}(\boldsymbol{s};J)=p^{-1}\left(  \frac{1-p^{-1}}{1-p^{-1-s_{\left(
N-1\right)  i}}}+p-2\right)  $.
\end{remark}

\begin{proof}
To compute $\boldsymbol{M}_{1}(\boldsymbol{s};J)$, we proceed as follows. We
set%
\[
T_{J}=\left\{  \left(  i,j\right)  \in T\times T;2\leq i<j\leq N-2,i,j\in
J\right\}
\]
as before, and for $\overline{\boldsymbol{a}}=\left(  \overline{a}_{i}\right)
_{i\in J}\in(\mathbb{F}_{p}^{\mathbb{\times}})^{\left\vert J\right\vert
}\smallsetminus\overline{\Pi}\left(  J\right)  $, with
\[
\overline{\Pi}\left(  J\right)  :=\left\{  \overline{\boldsymbol{a}}%
\in(\mathbb{F}_{p}^{\mathbb{\times}})^{\left\vert J\right\vert };\overline
{a}_{i}\neq\overline{a}_{j}\text{ if }i\neq j\text{, for }i\text{, }j\in
J\text{ and }\overline{a}_{s}\neq1\text{ for any }s\in J\right\}  ,
\]
we define%
\[
K(\overline{\boldsymbol{a}})=\left\{  \left(  i,j\right)  \in T_{J}%
;\overline{a}_{i}=\overline{a}_{j}\right\}  \text{, \ }K^{\left(  1\right)
}(\overline{\boldsymbol{a}})=\left\{  \left(  i,j\right)  \in T_{J}%
;\overline{a}_{i}=\overline{a}_{j}=1\right\}  ,
\]
and
\[
K^{\left(  2\right)  }(\overline{\boldsymbol{a}})=\left\{  i\in J;\overline
{a}_{i}=1\text{ and }\overline{a}_{i}\neq\overline{a}_{s}\text{ \ for any
}\left(  i,s\right)  \in T_{J}\right\}  .
\]
Notice that $K^{\left(  1\right)  }(\overline{\boldsymbol{a}})\subseteq
K(\overline{\boldsymbol{a}})$ and $K^{\left(  2\right)  }(\overline
{\boldsymbol{a}})\cap K_{\text{list}}(\overline{\boldsymbol{a}})=\varnothing.$
Now, we introduce on $(\mathbb{F}_{p}^{\mathbb{\times}})^{\left\vert
J\right\vert }\smallsetminus\overline{\Pi}\left(  J\right)  $, the following
equivalence relation:%
\[
\overline{\boldsymbol{a}}\sim\overline{\boldsymbol{b}}\text{ }\Leftrightarrow
K(\overline{\boldsymbol{a}})=K(\overline{\boldsymbol{b}})\text{ and
}K^{\left(  1\right)  }(\overline{\boldsymbol{a}})=K^{\left(  1\right)
}(\overline{\boldsymbol{b}})\text{ and }K^{\left(  2\right)  }(\overline
{\boldsymbol{a}})=K^{\left(  2\right)  }(\overline{\boldsymbol{b}}).
\]
We denote by $\overline{A}(\overline{\boldsymbol{a}})=\left\{  \overline
{\boldsymbol{b}}\in(\mathbb{F}_{p}^{\mathbb{\times}})^{\left\vert J\right\vert
}\smallsetminus\overline{\Pi}\left(  J\right)  ;\overline{\boldsymbol{a}}%
\sim\overline{\boldsymbol{b}}\right\}  $, the equivalence class defined by
$\overline{\boldsymbol{a}}\in(\mathbb{F}_{p}^{\mathbb{\times}})^{\left\vert
J\right\vert }\smallsetminus\overline{\Pi}\left(  J\right)  $. By taking a
unique representative in each equivalence class, we obtain $\mathcal{R}%
(J)\subset(\mathbb{F}_{p}^{\mathbb{\times}})^{\left\vert J\right\vert
}\smallsetminus\overline{\Pi}\left(  J\right)  $ such that
\begin{equation}
(\mathbb{F}_{p}^{\mathbb{\times}})^{\left\vert J\right\vert }=\bigsqcup
\limits_{\overline{\boldsymbol{a}}\in\mathcal{R}(J)}\overline{A}%
(\overline{\boldsymbol{a}})\bigsqcup\overline{\Pi}\left(  J\right)  .
\label{partition}%
\end{equation}
Given a subset $K\subseteq T_{J}$ with $K=\left\{  \left(  i_{1},j_{1}\right)
,\ldots,\left(  i_{m},j_{m}\right)  \right\}  $, we define $K_{\text{list}%
}=\left\{  i_{1},j_{1},\ldots,i_{m},j_{m}\right\}  \subseteq J$ as before.
With this notation, $\boldsymbol{M}_{1}(\boldsymbol{s};J)$\ equals
\begin{align}
&  \sum\limits_{\overline{\boldsymbol{a}}\in\mathcal{R}(J)}\text{ }%
\sum\limits_{\overline{\boldsymbol{b}}\in\overline{A}(\overline{\boldsymbol{a}%
})}\text{ }%
{\displaystyle\int\limits_{\boldsymbol{b}+(p\mathbb{Z}_{p})^{\left\vert
J\right\vert }}}
\text{ }%
{\displaystyle\prod\limits_{i\in J}}
\left\vert 1-x_{i}\right\vert _{p}^{s_{(N-1)i}}%
{\displaystyle\prod\limits_{\substack{2\leq i<j\leq N-2\\i,j\in J}}}
\left\vert x_{i}-x_{j}\right\vert _{p}^{s_{ij}}%
{\displaystyle\prod\limits_{i\in J}}
dx_{i}\label{Formula_M_1}\\
&  +\sum\limits_{\overline{\boldsymbol{b}}\in\overline{\Pi}\left(  J\right)
}\text{ }%
{\displaystyle\int\limits_{\boldsymbol{b}+(p\mathbb{Z}_{p})^{\left\vert
J\right\vert }}}
\text{ }%
{\displaystyle\prod\limits_{i\in J}}
\left\vert 1-x_{i}\right\vert _{p}^{s_{(N-1)i}}%
{\displaystyle\prod\limits_{\substack{2\leq i<j\leq N-2\\i,j\in J}}}
\left\vert x_{i}-x_{j}\right\vert _{p}^{s_{ij}}%
{\displaystyle\prod\limits_{i\in J}}
dx_{i}\nonumber\\
&  :=\boldsymbol{M}_{11}(\boldsymbol{s};J)+\boldsymbol{M}_{12}(\boldsymbol{s}%
;J).\nonumber
\end{align}
We now use that for each $\overline{\boldsymbol{a}}\in(\mathbb{F}%
_{p}^{\mathbb{\times}})^{\left\vert J\right\vert }\smallsetminus\overline{\Pi
}\left(  J\right)  $,
\[
T_{J}=K(\overline{\boldsymbol{a}})%
{\textstyle\bigsqcup}
\left\{  \left(  i,j\right)  \in T_{J};\overline{a}_{i}\neq\overline{a}%
_{j}\right\}
\]
and%
\[
J=K_{\text{list}}^{\left(  1\right)  }(\overline{\boldsymbol{a}})%
{\textstyle\bigsqcup}
K^{\left(  2\right)  }(\overline{\boldsymbol{a}})%
{\textstyle\bigsqcup}
\left\{  i\in J;\overline{a}_{i}\neq\overline{1}\right\}  ,
\]
to obtain
\[%
{\displaystyle\prod\limits_{i\in J}}
\left\vert 1-x_{i}\right\vert _{p}^{s_{(N-1)i}}=%
{\displaystyle\prod\limits_{i\in K_{\text{list}}^{\left(  1\right)
}(\overline{\boldsymbol{a}})}}
\left\vert 1-x_{i}\right\vert _{p}^{s_{(N-1)i}}%
{\displaystyle\prod\limits_{i\in K^{\left(  2\right)  }(\overline
{\boldsymbol{a}})}}
\left\vert 1-x_{i}\right\vert _{p}^{s_{(N-1)i}}%
\]
on $\boldsymbol{b}+(p\mathbb{Z}_{p})^{\left\vert J\right\vert }$, and
\[%
{\displaystyle\prod\limits_{\substack{2\leq i<j\leq N-2\\i,j\in J}}}
\left\vert x_{i}-x_{j}\right\vert _{p}^{s_{ij}}=%
{\displaystyle\prod\limits_{\left(  i,j\right)  \in K(\overline{\boldsymbol{a}%
})}}
\left\vert x_{i}-x_{j}\right\vert _{p}^{s_{ij}}%
\]
on $\boldsymbol{b}+(p\mathbb{Z}_{p})^{\left\vert J\right\vert }$. With
$J(\overline{\boldsymbol{a}}):=K^{\left(  2\right)  }(\overline{\boldsymbol{a}%
})\bigsqcup K_{\text{list}}(\overline{\boldsymbol{a}})$, we have
\begin{gather}
\boldsymbol{M}_{11}(\boldsymbol{s};J)=\sum\limits_{\overline{\boldsymbol{a}%
}\in\mathcal{R}(J)}\left\vert \overline{A}(\overline{\boldsymbol{a}%
})\right\vert p^{-\left\vert J\right\vert -\sum_{i\in K_{\text{list}}^{\left(
1\right)  }(\overline{\boldsymbol{a}})\sqcup K^{\left(  2\right)  }%
(\overline{\boldsymbol{a}})}s_{\left(  N-1\right)  i}-\sum_{\left(
i,j\right)  \in K(\overline{\boldsymbol{a}})}s_{ij}}\times\label{Formula_M_2}%
\\%
{\displaystyle\int\limits_{(\mathbb{Z}_{p})^{\left\vert J(\overline
{\boldsymbol{a}})\right\vert }}}
{\displaystyle\prod\limits_{i\in K_{\text{list}}^{\left(  1\right)
}(\overline{\boldsymbol{a}})\sqcup K^{\left(  2\right)  }(\overline
{\boldsymbol{a}})}}
\left\vert x_{i}\right\vert _{p}^{s_{(N-1)i}}%
{\displaystyle\prod\limits_{\left(  i,j\right)  \in K(\overline{\boldsymbol{a}%
})}}
\left\vert x_{i}-x_{j}\right\vert _{p}^{s_{ij}}%
{\displaystyle\prod\limits_{i\in J(\overline{\boldsymbol{a}})}}
dx_{i}\nonumber\\
=\sum\limits_{\overline{\boldsymbol{a}}\in\mathcal{R}(J)}\left\vert
\overline{A}(\overline{\boldsymbol{a}})\right\vert p^{-\left\vert J\right\vert
-\sum_{i\in K_{_{\text{list}}}^{\left(  1\right)  }(\overline{\boldsymbol{a}%
})\sqcup K^{\left(  2\right)  }(\overline{\boldsymbol{a}})}s_{\left(
N-1\right)  i}-\sum_{\left(  i,j\right)  \in K(\overline{\boldsymbol{a}}%
)}s_{ij}}\times\nonumber\\
\left\{
{\displaystyle\int\limits_{(\mathbb{Z}_{p})^{\left\vert K^{\left(  2\right)
}(\overline{\boldsymbol{a}})\right\vert }}}
{\displaystyle\prod\limits_{i\in K^{\left(  2\right)  }(\overline
{\boldsymbol{a}})}}
\left\vert x_{i}\right\vert _{p}^{s_{(N-1)i}}%
{\displaystyle\prod\limits_{i\in K^{\left(  2\right)  }(\overline
{\boldsymbol{a}})}}
dx_{i}\right\}  L_{2}\left(  \boldsymbol{s};K_{\text{list}}(\overline
{\boldsymbol{a}}),K_{\text{list}}^{\left(  1\right)  }\left(  \overline
{\boldsymbol{a}}\right)  ,K(\overline{\boldsymbol{a}})\right)  .\nonumber
\end{gather}
Now, by using the partition of $K(\overline{\boldsymbol{a}})$ given in Lemma
\ref{Lemma_partition}, we obtain
\begin{gather}
L_{2}\left(  \boldsymbol{s};K_{\text{list}}(\overline{\boldsymbol{a}%
}),K_{\text{list}}^{\left(  1\right)  }\left(  \overline{\boldsymbol{a}%
}\right)  ,K(\overline{\boldsymbol{a}})\right)  =L_{2}\left(  \boldsymbol{s}%
;K_{\text{list}}^{\left(  1\right)  }(\overline{\boldsymbol{a}}%
),K_{\text{list}}^{\left(  1\right)  }\left(  \overline{\boldsymbol{a}%
}\right)  ,T_{K_{list}^{\left(  1\right)  }(\overline{\boldsymbol{a}})}\right)
\label{Formula_L_2}\\
\times%
{\textstyle\prod\limits_{\left(  i,j\right)  \in\mathcal{R}\left(
\overline{\boldsymbol{a}}\right)  \backslash K^{\left(  1\right)  }%
(\overline{\boldsymbol{a}})}}
L_{1}\left(  \boldsymbol{s};K_{\text{list}}(\left(  i,j\right)  ,\overline
{\boldsymbol{a}}),T_{K_{list}(\left(  i,j\right)  ,\overline{\boldsymbol{a}}%
)}\right) \nonumber
\end{gather}
with the convention that $L_{2}\left(  \boldsymbol{s},\varnothing
,\varnothing,\varnothing\right)  :=1$. Finally,%
\begin{equation}
\boldsymbol{M}_{12}(\boldsymbol{s};J)=\sum\limits_{\overline{\boldsymbol{b}%
}\in\overline{\Pi}\left(  J\right)  }%
{\displaystyle\int\limits_{\boldsymbol{b}+(p\mathbb{Z}_{p})^{\left\vert
J\right\vert }}}
{\displaystyle\prod\limits_{i\in J}}
dx_{i}=p^{-\left\vert J\right\vert }\left\vert \overline{\Pi}\left(  J\right)
\right\vert . \label{Formula_M_3}%
\end{equation}
Hence, formula (\ref{Formula_M_a_J}) follows from (\ref{Formula_M_1}%
)-(\ref{Formula_M_3}) by using Lemmas \ref{Lemma_L_1}-\ref{Lema_L_2} and
Remark \ref{Nota_Lemma_L_2}.
\end{proof}

\subsection{Computation of $Z_{1}\left(  \boldsymbol{s};I\right)  $}

\begin{proposition}
\label{Proposition_1}Let $I$ be a non-empty subset of $T$. Then, the integral%
\[
Z_{1}\left(  \boldsymbol{s};I\right)  =\left\{
\begin{array}
[c]{ccc}%
{\displaystyle\int\limits_{\mathbb{Z}_{p}^{\left\vert I\right\vert }}}
\frac{%
{\displaystyle\prod\limits_{\substack{2\leq i<j\leq N-2\\i,j\in I}}}
\left\vert x_{i}-x_{j}\right\vert _{p}^{s_{ij}}}{%
{\displaystyle\prod\limits_{i\in I}}
\left\vert x_{i}\right\vert _{p}^{2+s_{1i}+s_{\left(  N-1\right)  i}%
+\sum_{2\leq j\leq N-2,j\neq i}s_{ij}}}%
{\displaystyle\prod\limits_{i\in I}}
dx_{i} & \text{if} & \left\vert I\right\vert \geq2\\
&  & \\%
{\displaystyle\int\limits_{\mathbb{Z}_{p}}}
\frac{1}{\left\vert x_{i}\right\vert _{p}^{2+s_{1i}+s_{\left(  N-1\right)
i}+\sum_{2\leq j\leq N-2,j\neq i}s_{ij}}}dx_{i} & \text{if} & \left\vert
I\right\vert =1
\end{array}
\right.
\]
converges on the set
\begin{gather*}
\left\{  \left(  s_{ij}\right)  \in\mathbb{C}^{D};\operatorname{Re}\left(
s_{ij}\right)  >-1\text{ for }2\leq i<j\leq N-2,i,j\in I\right\}  \text{ }%
\cap\\
\text{ }\left\{  \left(  s_{ij}\right)  \in\mathbb{C}^{D};1+\operatorname{Re}%
(s_{1i}+s_{\left(  N-1\right)  i})+\sum_{2\leq j\leq N-2,j\neq i}%
\operatorname{Re}(s_{ij})<0\text{ for }i\in I\right\}  ,
\end{gather*}
which is an open and connected subset of $\mathbb{C}^{D}$. In addition,
$Z_{1}\left(  \boldsymbol{s};I\right)  $ admits an analytic continuation to
$\mathbb{C}^{D}$ as a rational function of the form
\begin{equation}
Z_{1}\left(  \boldsymbol{s};I\right)  =\frac{Q_{I,1}(\left\{  p^{-s_{ij}%
};i,j\in\left\{  1,\ldots,N-1\right\}  \right\}  )}{S_{1}(\boldsymbol{s}%
;I)S_{2}(\boldsymbol{s};I)S_{3}(\boldsymbol{s};I)S_{4}(\boldsymbol{s};I)},
\label{zeta_L_1}%
\end{equation}
where $Q_{I,1}(\left\{  p^{-s_{ij}};i,j\in\left\{  1,\ldots,N-1\right\}
\right\}  )$ denotes a polynomial with rational coefficients in the variables
$p^{-s_{ij}}$, $i$, $j\in\left\{  1,\ldots,N-1\right\}  $,
\[
S_{1}(\boldsymbol{s};I)=%
{\displaystyle\prod\limits_{J\in\mathcal{H}_{1}(I)}}
\left(  1-p^{\left\vert J\right\vert +\sum_{i\in J}(s_{1}{}_{i}+s_{\left(
N-1\right)  i})+\sum_{\substack{2\leq i<j\leq N-2\\i\in J}}s_{ij}%
+\sum_{\substack{2\leq i<j\leq N-2\\i\in T\smallsetminus J,j\in J}}s_{ij}%
}\right)  ,
\]
where $\mathcal{H}_{1}(I)$ is a family of non-empty subsets of $I$, with
$I\in\mathcal{H}_{1}(I)$,
\[
S_{2}(\boldsymbol{s};I):=%
{\displaystyle\prod\limits_{\substack{J\subseteq I\\J\neq\varnothing}}}
\prod\limits_{K\in\mathcal{H}_{2}\left(  J\right)  }\left(  1-p^{-\left(
\left\vert K\right\vert -1+\sum_{\substack{2\leq i<j\leq N-2\\i,j\in K}%
}s_{ij}\right)  }\right)  ^{e_{K}},
\]
where $\mathcal{H}_{2}\left(  J\right)  $ is a family of non-empty subsets of
$J$, with $J\in\mathcal{H}_{2}\left(  J\right)  $, and the $e_{K}$'s are
positive integers,%
\[
S_{3}(\boldsymbol{s};I):=%
{\displaystyle\prod\limits_{\substack{J\subseteq I\\J\neq\varnothing}}}
\prod\limits_{ij\in G_{J}^{\left(  0\right)  }}\left(  1-p^{-1-s_{ij}}\right)
,
\]
where $G_{J}^{\left(  0\right)  }$ is a non-empty subset $\left\{  2\leq
i<j\leq N-2,i,j\in J\right\}  $,%
\[
S_{4}(\boldsymbol{s};I):=\prod\limits_{i\in G_{I}^{\left(  1\right)  }}\left(
1-p^{1+s_{1i}+s_{\left(  N-1\right)  i}+\sum_{2\leq j\leq N-2,j\neq i}s_{ij}%
}\right)  ,
\]
where $G_{I}^{\left(  1\right)  }$ is a non-empty subset $\left\{  2\leq
i<j\leq N-2,i,j\in I\right\}  $.
\end{proposition}

\begin{proof}
By using the partition $\mathbb{Z}_{p}^{\left\vert I\right\vert }%
=(p\mathbb{Z}_{p})^{\left\vert I\right\vert }\sqcup S_{0}^{\left\vert
I\right\vert }$ as in the proof of Lemma \ref{Lemma_L_1}, and a change of
variables, we get%
\begin{align*}
Z_{1}\left(  \boldsymbol{s};I\right)   &  =\frac{%
{\displaystyle\int\limits_{S_{0}^{\left\vert I\right\vert }}}
\frac{%
{\displaystyle\prod\limits_{\substack{2\leq i<j\leq N-2 \\i,j\in I }}}
\left\vert x_{i}-x_{j}\right\vert _{p}^{s_{ij}}}{%
{\displaystyle\prod\limits_{i\in I}}
\left\vert x_{i}\right\vert _{p}^{2+s_{1i}+s_{\left(  N-1\right)  i}%
+\sum_{2\leq j\leq N-2,j\neq i}s_{ij}}}%
{\displaystyle\prod\limits_{i\in I}}
dx_{i}}{1-p^{\left\vert I\right\vert +\sum_{i\in I}(s_{1}{}_{i}+s_{\left(
N-1\right)  i})+\sum_{\substack{2\leq i<j\leq N-2 \\i\in I}}s_{ij}%
+\sum_{\substack{2\leq i<j\leq N-2 \\i\in T\smallsetminus I,j\in I}}s_{ij}}}\\
&  =:\frac{C_{0}\left(  \boldsymbol{s}\right)  }{1-p^{\left\vert I\right\vert
+\sum_{i\in I}(s_{1}{}_{i}+s_{\left(  N-1\right)  i})+\sum_{\substack{2\leq
i<j\leq N-2 \\i\in I}}s_{ij}+\sum_{\substack{2\leq i<j\leq N-2 \\i\in
T\smallsetminus I,j\in I}}s_{ij}}}.
\end{align*}
We now use the partition $S_{0}^{\left\vert I\right\vert }=\sqcup_{J\subseteq
I,J\neq\varnothing}S_{J}^{\left\vert I\right\vert }$ to obtain
\[
C_{0}\left(  \boldsymbol{s}\right)  =\sum_{_{J\subseteq I,J\neq\varnothing}%
}C_{0,J}\left(  \boldsymbol{s}\right)  ,
\]
where
\[
C_{0,J}\left(  \boldsymbol{s}\right)  :=%
{\displaystyle\int\limits_{S_{J}^{\left\vert I\right\vert }}}
\frac{%
{\displaystyle\prod\limits_{\substack{2\leq i<j\leq N-2 \\i,j\in I }}}
\left\vert x_{i}-x_{j}\right\vert _{p}^{s_{ij}}}{%
{\displaystyle\prod\limits_{i\in I}}
\left\vert x_{i}\right\vert _{p}^{2+s_{1i}+s_{\left(  N-1\right)  i}%
+\sum_{2\leq j\leq N-2,j\neq i}s_{ij}}}%
{\displaystyle\prod\limits_{i\in I}}
dx_{i},
\]
and consequently,
\[
Z_{1}\left(  \boldsymbol{s};I\right)  =\frac{C_{0,I}\left(  \boldsymbol{s}%
\right)  +\sum_{_{J\subsetneqq I,J\neq\varnothing}}C_{0,J}\left(
\boldsymbol{s}\right)  }{1-p^{\left\vert I\right\vert +\sum_{i\in I}(s_{1}%
{}_{i}+s_{\left(  N-1\right)  i})+\sum_{\substack{2\leq i<j\leq N-2 \\i\in
I}}s_{ij}+\sum_{\substack{2\leq i<j\leq N-2 \\i\in T\smallsetminus I,j\in
I}}s_{ij}}}.
\]
On the other hand, by using (\ref{Calcul_abs}), we have $C_{0,I}\left(
\boldsymbol{s}\right)  =L_{0}\left(  \boldsymbol{s};I\right)  $, and if
$J\subsetneqq I$,%
\[
C_{0,J}\left(  \boldsymbol{s}\right)  =\left\{
{\displaystyle\int\limits_{(p\mathbb{Z}_{p})^{\left\vert I\smallsetminus
J\right\vert }}}
\frac{%
{\displaystyle\prod\limits_{\substack{2\leq i<j\leq N-2 \\i,j\in
I\smallsetminus J}}}
\left\vert x_{i}-x_{j}\right\vert _{p}^{s_{ij}}}{%
{\displaystyle\prod\limits_{i\in I\smallsetminus J}}
\left\vert x_{i}\right\vert _{p}^{2+s_{1i}+s_{\left(  N-1\right)  i}%
+\sum_{\substack{2\leq j\leq N-2 \\j\neq i}}s_{ij}}}%
{\displaystyle\prod\limits_{i\in I\smallsetminus J}}
dx_{i}\right\}  L_{0}\left(  \boldsymbol{s};J\right)
\]%
\[%
\begin{array}
[c]{c}%
=p^{\left\vert I\smallsetminus J\right\vert +\sum_{\text{ }i\in
I\smallsetminus J}(s_{1}{}_{i}+s_{\left(  N-1\right)  i})+\sum
_{\substack{2\leq i<j\leq N-2 \\i\in I\smallsetminus J\text{ }}}s_{ij}%
+\sum_{\substack{2\leq i<j\leq N-2 \\i\in T\diagdown\left(  I\smallsetminus
J\right)  ,j\in I\smallsetminus J}}s_{ij}}\times\\
\left\{
{\displaystyle\int\limits_{\mathbb{Z}_{p}^{\left\vert I\smallsetminus
J\right\vert }}}
\frac{%
{\displaystyle\prod\limits_{\substack{2\leq i<j\leq N-2 \\i,j\in
I\smallsetminus J}}}
\left\vert x_{i}-x_{j}\right\vert _{p}^{s_{ij}}}{%
{\displaystyle\prod\limits_{i\in I\smallsetminus J}}
\left\vert x_{i}\right\vert _{p}^{2+s_{1i}+s_{\left(  N-1\right)  i}%
+\sum_{\substack{2\leq j\leq N-2 \\j\neq i}}s_{ij}}}%
{\displaystyle\prod\limits_{i\in I\smallsetminus J}}
dx_{i}\right\}  L_{0}\left(  \boldsymbol{s};J\right) \\
=p^{\left\vert I\smallsetminus J\right\vert +\sum_{\text{ }i\in
I\smallsetminus J}(s_{1}{}_{i}+s_{\left(  N-1\right)  i})+\sum
_{\substack{2\leq i<j\leq N-2 \\i\in I\smallsetminus J}}s_{ij}+\sum
_{\substack{2\leq i<j\leq N-2 \\i\in T\diagdown\left(  I\smallsetminus
J\right)  ,j\in I\smallsetminus J}}s_{ij}}\times\\
Z_{1}\left(  \boldsymbol{s};I\smallsetminus J\right)  L_{0}\left(
\boldsymbol{s};J\right)  .
\end{array}
\]
Therefore
\begin{equation}
Z_{1}\left(  \boldsymbol{s};I\right)  =\frac{L_{0}\left(  \boldsymbol{s}%
;I\right)  +\sum\limits_{J\subsetneqq I,J\neq\varnothing}p^{M(\boldsymbol{s}%
,J)}Z_{1}\left(  \boldsymbol{s};I\smallsetminus J\right)  L_{0}\left(
\boldsymbol{s};J\right)  }{1-p^{\left\vert I\right\vert +\sum_{i\in I}(s_{1}%
{}_{i}+s_{\left(  N-1\right)  i})+\sum_{\substack{2\leq i<j\leq N-2 \\i\in
I}}s_{ij}+\sum_{\substack{2\leq i<j\leq N-2 \\i\in T\smallsetminus I,j\in
I}}s_{ij}}}, \label{Int2}%
\end{equation}
where
\begin{align*}
M(\boldsymbol{s},J)  &  :=\left\vert I\smallsetminus J\right\vert
+\sum_{\text{ }i\in I\smallsetminus J}(s_{1}{}_{i}+s_{\left(  N-1\right)
i})+\sum_{\substack{2\leq i<j\leq N-2 \\i\in I\smallsetminus J}}s_{ij}\\
&  +\sum_{\substack{2\leq i<j\leq N-2 \\i\in T\diagdown\left(  I\smallsetminus
J\right)  ,j\in I\smallsetminus J}}s_{ji}.
\end{align*}
Notice that in (\ref{Int2}), $Z_{1}\left(  \boldsymbol{s};I\smallsetminus
J\right)  $ may occur with $\left\vert I\smallsetminus J\right\vert =1$, say
$I\smallsetminus J=\left\{  i\right\}  $, in this case $Z_{1}\left(
\boldsymbol{s};I\right)  $ becomes
\begin{equation}%
{\displaystyle\int\limits_{\mathbb{Z}_{p}}}
\frac{1}{\left\vert x_{i}\right\vert _{p}^{2+s_{1i}+s_{\left(  N-1\right)
i}+\sum_{2\leq j\leq N-2,j\neq i}s_{ij}}}dx_{i}=\frac{1-p^{-1}}{1-p^{1+s_{1i}%
+s_{\left(  N-1\right)  i}+\sum_{2\leq j\leq N-2,j\neq i}s_{ij}}}
\label{special_case}%
\end{equation}
for $\operatorname{Re}(s_{1i})+\operatorname{Re}(s_{\left(  N-1\right)
i})+\sum_{2\leq j\leq N-2,j\neq i}\operatorname{Re}(s_{ij})<-1$.

Finally, formula (\ref{Int2}) gives a recursive algorithm for computing
$Z_{1}\left(  \boldsymbol{s};I\right)  $, since $I\smallsetminus J\subsetneqq
I\subseteq T$ and $L_{0}\left(  \boldsymbol{s};I\right)  $, $L_{0}\left(
\boldsymbol{s};J\right)  $ can be effectively computed, see Corollary
\ref{Cor_Lemma_L_1}, by using this algorithm and (\ref{special_case}), we
obtain (\ref{zeta_L_1}).
\end{proof}

\begin{remark}
\label{Remark_Lemma_3}Given positive integers $N_{i}$, $i\in I\subseteq T$,
$v$, and complex numbers $s_{i}$ for $i\in I$, we notice that the function
$\frac{1}{1-p^{-v-\sum_{i\in I}N_{i}s_{i}}}$ gives rise to a holomorphic
function of the $s_{i}$ on the half-plane $\sum_{i\in I}N_{i}\operatorname{Re}%
\left(  s_{i}\right)  +v>0$. As a consequence of Proposition
\ref{Proposition_1} there exist families $\mathfrak{F}_{1}$, $\mathfrak{F}%
_{2}$ of non-empty subsets of $T$, and a non-empty subset $\mathcal{G}$ of
$\left\{  ij;2\leq i<j\leq N-2,i,j\in T\right\}  $, such that all\ the
integrals $Z_{1}\left(  \boldsymbol{s};I\right)  $ for all $I\subseteq T$ are
holomorphic functions of $\boldsymbol{s}$ on the solution set of the
conditions:%
\begin{gather}
\left\vert J\right\vert +\sum_{i\in J}(\operatorname{Re}\left(  s_{1}{}%
_{i}\right)  +\operatorname{Re}\left(  s_{\left(  N-1\right)  i}\right)
)+\sum_{\substack{2\leq i<j\leq N-2 \\i\in J}}\operatorname{Re}\left(
s_{ij}\right) \tag{$C1$}\\
+\sum_{\substack{2\leq i<j\leq N-2 \\i\in T\smallsetminus J,j\in
J}}\operatorname{Re}\left(  s_{ij}\right)  <0\text{ for }J\in\mathfrak{F}%
_{1};\nonumber
\end{gather}%
\begin{equation}
\left\vert K\right\vert -1+\sum_{\substack{2\leq i<j\leq N-2 \\i,j\in
K}}\operatorname{Re}(s_{ij})>0\text{ for }K\in\mathfrak{F}_{2}\text{;}
\tag{$C2$}%
\end{equation}%
\begin{equation}
1+\operatorname{Re}(s_{ij})>0\text{ for }ij\in\mathcal{G\subseteq}\left\{
ij;2\leq i<j\leq N-2\right\}  . \tag{$C3$}%
\end{equation}
Notice that the condition
\[
1+\operatorname{Re}(s_{1i})+\operatorname{Re}(s_{\left(  N-1\right)  i}%
)+\sum_{2\leq j\leq N-2,j\neq i}\operatorname{Re}(s_{ij})<0
\]
is included in Condition C1 taking $\left\vert J\right\vert =1$. This fact
follows from the following identities:%
\begin{gather*}
\sum_{\substack{2\leq i<j\leq N-2 \\i\in J,j\in T}}s_{ij}+\sum
_{\substack{2\leq i<j\leq N-2 \\i\in T\smallsetminus J,j\in J}}s_{ij}%
=\sum_{\substack{2\leq i<j\leq N-2 \\i\in J,j\in T}}s_{ij}+\sum
_{\substack{2\leq i<j\leq N-2 \\i\in T,j\in J}}s_{ij}-\sum_{\substack{2\leq
i<j\leq N-2 \\i,j\in J }}s_{ij}=\\
\sum_{\substack{2\leq i<j\leq N-2 \\i\in J,j\in T\smallsetminus J}}s_{ij}%
+\sum_{\substack{2\leq i<j\leq N-2 \\i\in T,j\in J}}s_{ij}=\sum
_{\substack{2\leq i<j\leq N-2 \\i\in J,j\in T\smallsetminus J}}s_{ij}%
+\sum_{\substack{2\leq i<j\leq N-2 \\i,j\in J}}s_{ij}+\sum_{\substack{2\leq
i<j\leq N-2 \\i\in T\smallsetminus J,j\in J}}s_{ij}=\\
\sum_{\substack{2\leq i<j\leq N-2 \\i\in J,j\in T\smallsetminus J}}s_{ij}%
+\sum_{\substack{2\leq j<i\leq N-2 \\i\in J,j\in T\smallsetminus J}%
}s_{ij}+\sum_{\substack{2\leq i<j\leq N-2 \\i,j\in J}}s_{ij}=\sum
_{\substack{2\leq j\leq N-2 \\j\neq i,i\in J,j\in T\smallsetminus J}%
}s_{ij}+\sum_{\substack{2\leq i<j\leq N-2 \\i,j\in J}}s_{ij}.
\end{gather*}
Finally, by taking $J=\left\{  i\right\}  $, the last formula becomes
$\sum_{\substack{2\leq j\leq N-2 \\j\neq i}}s_{ij}$.
\end{remark}

Denote by $D_{I,1}$ the natural domain of definition of $Z_{1}\left(
\boldsymbol{s};I\right)  $, i.e. $D_{I,1}$\ is an open and connected subset of
$\mathbb{C}^{D}$ in which $Z_{1}\left(  \boldsymbol{s};I\right)  $ is
holomorphic and there no exists a larger domain where this property holds.

\begin{lemma}
\label{Lemma_3A}Take $I$ to be a non-empty subset of $T$ and set
$H_{I,1}(\mathbb{C})$ to be the solution set in $\mathbb{C}^{D}$ of the
following conditions:
\begin{equation}
1+\operatorname{Re}\left(  s_{1i}\right)  +\operatorname{Re}\left(  s_{\left(
N-1\right)  i}\right)  +\sum_{2\leq j\leq N-2,j\neq i}\operatorname{Re}\left(
s_{ij}\right)  <0\text{, for }i\in I. \label{conditions}%
\end{equation}
Then $D_{I,1}$ is contained in $H_{I,1}(\mathbb{C})$.
\end{lemma}

\begin{proof}
Denote by $H_{I,1}(\mathbb{R})$ the solution set of (\ref{conditions}) in
\ $\mathbb{R}^{D}$. Set $\operatorname{Re}\left(  D_{I,1}\right)  =\left\{
\operatorname{Re}(s_{ij})\in\mathbb{R}^{D};(s_{ij})\in D_{I,1}\right\}  $.
With this notation, it is sufficient to show that $\operatorname{Re}\left(
D_{I,1}\right)  \subset H_{I,1}(\mathbb{R})$. In order to do this, we show
that the integral $Z_{1}\left(  \widetilde{\boldsymbol{s}};I\right)  $
diverges to $+\infty$ for any $\widetilde{\boldsymbol{s}}\in\mathbb{R}%
^{D}\smallsetminus H_{I,1}(\mathbb{R})$. We prove this last assertion by
contradiction. Assume that $Z_{1}\left(  \widetilde{\boldsymbol{s}};I\right)
<+\infty$ for $\widetilde{\boldsymbol{s}}=\left(  \widetilde{s}_{ij}\right)
\in\mathbb{R}^{D}$ with $\widetilde{s}_{ij}\geq0$ for $2\leq i<j\leq N-2$,
$i$, $j\in I$ and that $\widetilde{\boldsymbol{s}}\notin H_{I,1}(\mathbb{R})$.
This last condition implies that at least a condition of the form%
\begin{equation}
1+\widetilde{s}_{1i_{0}}+\widetilde{s}_{\left(  N-1\right)  i_{0}}+\sum_{2\leq
j\leq N-2,j\neq i_{0}}\widetilde{s}_{ij}\geq0 \label{condition1}%
\end{equation}
for some $i_{0}\in I$, holds. Then, from $Z_{1}\left(  \widetilde
{\boldsymbol{s}};I\right)  <+\infty$, we have
\[
I(\widetilde{\boldsymbol{s}};A):=%
{\displaystyle\int\limits_{A}}
\frac{%
{\displaystyle\prod\limits_{\substack{2\leq i<j\leq N-2 \\i,j\in I}}}
\left\vert x_{i}-x_{j}\right\vert _{p}^{\widetilde{s}_{ij}}}{%
{\displaystyle\prod\limits_{i\in I}}
\left\vert x_{i}\right\vert _{p}^{2+\widetilde{s}_{1i}+\widetilde{s}_{\left(
N-1\right)  i}+\sum_{2\leq j\leq N-2,j\neq i}\widetilde{s}_{ij}}}%
{\displaystyle\prod\limits_{i\in I}}
dx_{i}<+\infty
\]
for any $A\subset\mathbb{Z}_{p}^{\left\vert I\right\vert }$. Take%
\[
A_{0}=\left\{  \left(  x_{i}\right)  _{i\in I}\in\mathbb{Z}_{p}^{\left\vert
I\right\vert };\left\vert x_{i_{0}}\right\vert _{p}<1\text{ and }\left\vert
x_{i}\right\vert _{p}=1\text{ for }i\in I\smallsetminus\left\{  i_{0}\right\}
\right\}  .
\]
Then, by (\ref{condition1}) and some $\epsilon\geq0$,
\[
I(\widetilde{\boldsymbol{s}};A_{0})=%
{\displaystyle\int\limits_{A_{0}}}
\frac{%
{\displaystyle\prod\limits_{\substack{2\leq i<j\leq N-2 \\i,j\in
I\smallsetminus\left\{  i_{0}\right\}  }}}
\left\vert x_{i}-x_{j}\right\vert _{p}^{\widetilde{s}_{ij}}}{\left\vert
x_{i_{0}}\right\vert _{p}^{1+\epsilon}}%
{\displaystyle\prod\limits_{i\in I}}
dx_{i}=+\infty.
\]
Therefore, if $Z_{1}\left(  \widetilde{\boldsymbol{s}};I\right)  <+\infty$,
necessarily $\widetilde{\boldsymbol{s}}\in H_{I,1}(\mathbb{R})$.
\end{proof}

\begin{corollary}
\label{Corollary_Lemma_3A}If $\boldsymbol{s}=\left(  s_{ij}\right)
\in\mathbb{R}^{D}$, with $s_{ij}\geq0$ for $i,j\in\left\{  1,\ldots
,N-1\right\}  $, then the integral $Z_{1}\left(  \boldsymbol{s};I\right)  $
diverges to $+\infty$, for any non-empty subset $I$ of $T$.
\end{corollary}

\subsection{Computation of $Z_{0}\left(  \boldsymbol{s};I\right)  $}

\begin{proposition}
\label{Proposition_2}Let $I$ be a subset of $T$ satisfying $\left\vert
I\right\vert \geq2$. Then, the integral%
\[
Z_{0}\left(  \boldsymbol{s};I\right)  =%
{\displaystyle\int\limits_{\mathbb{Z}_{p}^{\left\vert I\right\vert }}}
{\displaystyle\prod\limits_{i\in I}}
\left\vert x_{i}\right\vert _{p}^{s_{1i}}\left\vert 1-x_{i}\right\vert
_{p}^{s_{(N-1)i}}\text{ }%
{\displaystyle\prod\limits_{\substack{2\leq i<j\leq N-2 \\i,j\in I}}}
\left\vert x_{i}-x_{j}\right\vert _{p}^{s_{i}{}_{j}}%
{\displaystyle\prod\limits_{i\in I}}
dx_{i}%
\]
gives rise to a holomorphic function on
\begin{gather*}
H_{I,0}:=\left\{  \left(  s_{ij}\right)  \in\mathbb{C}^{D};\operatorname{Re}%
\left(  s_{ij}\right)  >0\text{ for }i,j\in I\right\}  \text{ }\cap\left\{
\left(  s_{ij}\right)  \in\mathbb{C}^{D};\operatorname{Re}(s_{1i})>0\text{ for
}i\in I\right\} \\
\cap\left\{  \left(  s_{ij}\right)  \in\mathbb{C}^{D};\operatorname{Re}%
(s_{\left(  N-1\right)  i})>0\text{ for }i\in I\right\}  ,
\end{gather*}
which is an open and connected subset of $\mathbb{C}^{D}$. Furthermore
$Z_{0}\left(  \boldsymbol{s};I\right)  $ has an analytic continuation as a
rational function of the form
\[
Z_{0}\left(  \boldsymbol{s};I\right)  =\frac{Q_{I,0}(\left\{  p^{-s_{1i}%
},p^{-s_{(N-1)i}},p^{-s_{ij}};i,j\in T\right\}  )}{\prod\limits_{i=0}^{2}%
R_{i}(\boldsymbol{s};I,I)%
{\textstyle\prod\nolimits_{i=1}^{3}}
U_{i}(\boldsymbol{s};I)},
\]
where $Q_{I,0}(\left\{  p^{-s_{1i}},p^{-s_{(N-1)i}},p^{-s_{ij}};i,j\in
T\right\}  $ is a polynomial in the variables $p^{-s_{1i}}$, $p^{-s_{(N-1)i}}%
$, $p^{-s_{ij}}$ for $i$, $j\in T$, $U_{i}(\boldsymbol{s};I)$, $i=1$, $2$, $3$
are as in Lemma \ref{Lemma_M_a_J},
\[
R_{1}(\boldsymbol{s};I,I)=\left(  1-p^{-1-s_{1i}}\right)  ^{h_{1}},
\]%
\[
R_{2}(\boldsymbol{s};I,K)=\prod\limits_{\left(  J,R\right)  \in\mathcal{G}%
_{2}\left(  I\times I\right)  }\left(  1-p^{-\left\vert J\right\vert
-\sum_{i\in R}s_{1i}-\sum\nolimits_{2\leq i<j\leq N-2,\text{ }i,j\in J}s_{ij}%
}\right)  ^{l_{\left(  J,R\right)  }},
\]
$R_{0}(\boldsymbol{s};I,I)$, $\mathcal{G}_{2}\left(  I\times I\right)  $ are
as in Lemma \ref{Lema_L_2}, and the $l_{\left(  J,R\right)  }$'s\ are positive integers.
\end{proposition}

\begin{proof}
By using that $\mathbb{Z}_{p}^{\left\vert I\right\vert }=(p\mathbb{Z}%
_{p})^{\left\vert I\right\vert }\sqcup S_{0}^{\left\vert I\right\vert }$, we
have
\begin{equation}
Z_{0}\left(  \boldsymbol{s};I\right)  =V_{1}\left(  \boldsymbol{s};I\right)
+V_{2}\left(  \boldsymbol{s};I\right)  , \label{Formula_Z_I_0}%
\end{equation}
where
\[
V_{1}\left(  \boldsymbol{s};I\right)  :=%
{\displaystyle\int\limits_{(p\mathbb{Z}_{p})^{\left\vert I\right\vert }}}
{\displaystyle\prod\limits_{i\in I}}
\left\vert x_{i}\right\vert _{p}^{s_{1i}}\text{ }%
{\displaystyle\prod\limits_{\substack{2\leq i<j\leq N-2 \\i,j\in I}}}
\left\vert x_{i}-x_{j}\right\vert _{p}^{s_{i}{}_{j}}%
{\displaystyle\prod\limits_{i\in I}}
dx,
\]%
\[
V_{2}\left(  \boldsymbol{s};I\right)  :=%
{\displaystyle\int\limits_{S_{0}^{\left\vert I\right\vert }}}
{\displaystyle\prod\limits_{i\in I}}
\left\vert x_{i}\right\vert _{p}^{s_{1i}}\left\vert 1-x_{i}\right\vert
_{p}^{s_{(N-1)i}}\text{ }%
{\displaystyle\prod\limits_{\substack{2\leq i<j\leq N-2 \\i,j\in I}}}
\left\vert x_{i}-x_{j}\right\vert _{p}^{s_{i}{}_{j}}%
{\displaystyle\prod\limits_{i\in I}}
dx.
\]
Now, by changing variables and using Lemma \ref{Lema_L_2} with $t=1$,
$V_{1}\left(  \boldsymbol{s};I\right)  $ equals
\begin{gather}
p^{-\left\vert I\right\vert -\sum_{i\in I}s_{1i}-\sum_{\substack{2\leq i<j\leq
N-2 \\i,j\in I}}s_{ij}}%
{\displaystyle\int\limits_{\mathbb{Z}_{p}^{\left\vert I\right\vert }}}
{\displaystyle\prod\limits_{i\in I}}
\left\vert x_{i}\right\vert _{p}^{s_{1i}}%
{\displaystyle\prod\limits_{\substack{2\leq i<j\leq N-2 \\i,j\in I}}}
\left\vert x_{i}-x_{j}\right\vert _{p}^{s_{ij}}%
{\displaystyle\prod\limits_{i\in I}}
dx_{i}\label{Formula_Z_I_0_1}\\
=p^{-\left\vert I\right\vert -\sum_{i\in I}s_{1i}-\sum_{\substack{2\leq
i<j\leq N-2 \\i,j\in I}}s_{ij}}L_{2}\left(  \boldsymbol{s};I,I,T_{I}\right)
.\nonumber
\end{gather}
To compute $V_{2}\left(  \boldsymbol{s};I\right)  $, we use the partition
$S_{0}^{\left\vert I\right\vert }=\sqcup_{J\subseteq I,J\neq\varnothing}%
S_{J}^{\left\vert I\right\vert }$, with $S_{J}^{\left\vert I\right\vert
}=\left\{  \left(  x_{i}\right)  _{i\in I}\in\mathbb{Z}_{p}^{\left\vert
I\right\vert };\left\vert x_{i}\right\vert _{p}=1\Leftrightarrow i\in
J\right\}  $, then $V_{2}\left(  \boldsymbol{s};I\right)  $ equals
\begin{align}
&  \sum_{\substack{J\subseteq I \\J\neq\varnothing}}%
{\displaystyle\int\limits_{S_{J}^{\left\vert I\right\vert }}}
{\displaystyle\prod\limits_{i\in I}}
\left\vert x_{i}\right\vert _{p}^{s_{1i}}\left\vert 1-x_{i}\right\vert
_{p}^{s_{(N-1)i}}\text{ }%
{\displaystyle\prod\limits_{\substack{2\leq i<j\leq N-2 \\i,j\in I}}}
\left\vert x_{i}-x_{j}\right\vert _{p}^{s_{i}{}_{j}}%
{\displaystyle\prod\limits_{i\in I}}
dx\label{Formula_Z_I_0_2}\\
&  =\sum_{\substack{J\subseteq I \\J\neq\varnothing}}M_{J}\left(
\boldsymbol{s}\right)  ,\nonumber
\end{align}
where
\begin{gather*}
M_{J}\left(  \boldsymbol{s}\right)  =%
{\displaystyle\int\limits_{S_{J}^{\left\vert I\right\vert }}}
{\displaystyle\prod\limits_{i\in I\smallsetminus J}}
\left\vert x_{i}\right\vert _{p}^{s_{1i}}%
{\displaystyle\prod\limits_{\substack{2\leq i<j\leq N-2 \\i,j\in
I\smallsetminus J}}}
\left\vert x_{i}-x_{j}\right\vert _{p}^{s_{ij}}%
{\displaystyle\prod\limits_{i\in J}}
\left\vert 1-x_{i}\right\vert _{p}^{s_{(N-1)i}}\times\\%
{\displaystyle\prod\limits_{\substack{2\leq i<j\leq N-2 \\i,j\in J}}}
\left\vert x_{i}-x_{j}\right\vert _{p}^{s_{ij}}%
{\displaystyle\prod\limits_{i\in I}}
dx=%
{\displaystyle\int\limits_{(p\mathbb{Z}_{p})^{\left\vert I\smallsetminus
J\right\vert }}}
{\displaystyle\prod\limits_{i\in I\smallsetminus J}}
\left\vert x_{i}\right\vert _{p}^{s_{1i}}%
{\displaystyle\prod\limits_{\substack{2\leq i<j\leq N-2 \\i,j\in
I\smallsetminus J}}}
\left\vert x_{i}-x_{j}\right\vert _{p}^{s_{ij}}%
{\displaystyle\prod\limits_{i\in I\smallsetminus J}}
dx_{i}\\
\times%
{\displaystyle\int\limits_{(\mathbb{Z}_{p}^{\times})^{\left\vert J\right\vert
}}}
{\displaystyle\prod\limits_{i\in J}}
\left\vert 1-x_{i}\right\vert _{p}^{s_{(N-1)i}}%
{\displaystyle\prod\limits_{\substack{2\leq i<j\leq N-2 \\i,j\in J}}}
\left\vert x_{i}-x_{j}\right\vert _{p}^{s_{ij}}%
{\displaystyle\prod\limits_{i\in J}}
dx_{i}:=H_{0}\left(  \boldsymbol{s};I\smallsetminus J\right)  M_{1}\left(
\boldsymbol{s};J\right)  .
\end{gather*}
We notice that if $J=I$, then, by convention, $H_{0}\left(  \boldsymbol{s}%
;I\smallsetminus J\right)  =1$. Now suppose that $J\subsetneqq I$. From Lemma
\ref{Lema_L_2} with $t=1$, we have
\begin{equation}
H_{0}\left(  \boldsymbol{s};I\smallsetminus J\right)  =p^{-\left\vert
I\smallsetminus J\right\vert -\sum_{i\in I\smallsetminus J}s_{1i}-\sum_{
_{\substack{2\leq i<j\leq N-2 \\i,j\in I\smallsetminus J}}}s_{ij}}L_{2}\left(
\boldsymbol{s};I\smallsetminus J,I\smallsetminus J,T_{I\smallsetminus
J}\right)  . \label{Formula_Z_I_0_3}%
\end{equation}
The announced result follows from formulas (\ref{Formula_Z_I_0}%
)-(\ref{Formula_Z_I_0_3}), and $M_{1}\left(  \boldsymbol{s};J\right)  $ by
using Lemmas \ref{Lema_L_2}-\ref{Lemma_M_a_J} and Remark
\ref{Nota_Lemma_M_a_J}.
\end{proof}

\begin{remark}
\label{Remark_Lemma_5}As a consequence of Proposition \ref{Proposition_2} all
the integrals $Z_{0}\left(  \boldsymbol{s};I\right)  $ for all $I\subseteq T$
\ are holomorphic functions of $\boldsymbol{s}$ on the solution set in
$\mathbb{C}^{D}$ of the following conditions:%
\begin{equation}
\left\vert J\right\vert +\sum_{i\in S}\operatorname{Re}(s_{ti})+\sum
\nolimits_{2\leq i<j\leq N-2,\text{ }i,j\in J}\operatorname{Re}(s_{ij}%
)>0\text{ for }J\times S\in\mathfrak{F}_{3}\text{,} \tag{$C4$}%
\end{equation}
with $S\subseteq J$, $t\in\left\{  1,N-1\right\}  $,and $\mathfrak{F}_{3}$ a
family of non-empty subsets of $I\times I$;
\begin{equation}
\left\vert K\right\vert -1+\sum_{\substack{2\leq i<j\leq N-2 \\i,j\in
K}}\operatorname{Re}(s_{ij})>0\text{ for }K\in\mathfrak{F}_{4}\text{,}
\tag{$C5$}%
\end{equation}
where $\mathfrak{F}_{4}$ s a family of non-empty subsets of $I$;%
\begin{equation}
1+\operatorname{Re}(s_{ij})>0\text{ for }ij\in G_{T}, \tag{$C6$}%
\end{equation}
where $G_{T}$ is a non-empty subset of $\left\{  2\leq i<j\leq N-2,\text{
}i,j\in J\right\}  $ with $(N-1)i$, $1i\in G_{T}$.
\end{remark}

\begin{remark}
\label{Remark_B_Lemma_5}If $\boldsymbol{s}=\left(  0\right)  _{ij}$ for $i$,
$j\in\left\{  1,\ldots,N-1\right\}  $, then $Z_{0}\left(  \boldsymbol{0}%
;I\right)  =1$, for any non-empty subset $I$ of $T$.
\end{remark}

\begin{definition}
Denote by $H(\mathbb{R})$, respectively by $H(\mathbb{C})$, the solution set
of conditions C1-C6 in $\mathbb{R}^{D}$, respectively in $\mathbb{C}^{D}$.
\end{definition}

\subsection{Main Theorem}

\begin{lemma}
\label{Lemma_H_0}Consider the following conditions:%
\begin{gather}
\left\vert J\right\vert +\sum_{i\in J}(\operatorname{Re}\left(  s_{1}{}%
_{i}\right)  +\operatorname{Re}\left(  s_{\left(  N-1\right)  i}\right)
)+\sum_{\substack{2\leq i<j\leq N-2 \\i\in J}}\operatorname{Re}\left(
s_{ij}\right) \tag{$C1^\prime$}\\
+\sum_{\substack{2\leq i<j\leq N-2 \\i\in T\smallsetminus J,j\in
J}}\operatorname{Re}\left(  s_{ij}\right)  <0\text{ for }J\subseteq T\text{,
}\left\vert J\right\vert \geq1;\nonumber
\end{gather}%
\begin{equation}
\left\vert J\right\vert -1+\sum_{\substack{2\leq i<j\leq N-2 \\i,j\in
J}}\operatorname{Re}(s_{ij})>0\text{ for }J\subseteq T\text{, }\left\vert
J\right\vert \geq2\text{;} \tag{$C2^\prime$}%
\end{equation}%
\begin{gather}
\left\vert J\right\vert +\sum_{i\in S}\operatorname{Re}(s_{ti})+\sum
\nolimits_{2\leq i<j\leq N-2,\text{ }i,j\in J}\operatorname{Re}(s_{ij}%
)>0\text{ }\tag{$C3^\prime$}\\
\text{for }t\in\left\{  1,N-1\right\}  \text{, }J\times S\subseteq T\times
T\text{ with }\left\vert J\right\vert \geq2\text{ or }\left\vert S\right\vert
\geq1,S\subseteq J\text{;}\nonumber
\end{gather}%
\begin{equation}
1+\operatorname{Re}(s_{ij})>0\text{ for }ij\in\left\{  (i,j);1\leq i<j\leq
N-1\right\}  . \tag{$C4^\prime$}%
\end{equation}
Denote by $H_{0}(\mathbb{R})$, respectively by $H_{0}(\mathbb{C})$, the
solution set of conditions $C1^{\prime}$-$C4^{\prime}$ in $\mathbb{R}^{D}$,
respectively in $\mathbb{C}^{D}$. Then $H_{0}(\mathbb{R})$ is convex and
bounded set with non-empty interior, and $H_{0}(\mathbb{C})$ contains an open
and connected subset of $\mathbb{C}^{D}$. Furthermore, $H_{0}(\mathbb{R}%
)\subset H(\mathbb{R})$ and $H_{0}(\mathbb{C})\subset H(\mathbb{C})$.
\end{lemma}

\begin{proof}
We first notice that for all $N\geq4$, the solution set $H_{0}(\mathbb{R})$ is
an open convex set because it is a finite intersection of open half-spaces.

\textbf{Claim }$H_{0}(\mathbb{R})$ is a non-empty bounded subset\textbf{. }We
consider the case\textbf{\ }$N\geqslant5$ in which $|T|\geqslant2$. Set
$N_{1}=\frac{(N-4)(N-3)}{2}$. We define, for $i,j\in\left\{
2,...,N-2\right\}  $, the following conditions:%
\begin{gather}
-\frac{2}{3N_{1}}<\operatorname{Re}(s_{ij})<0\text{ },\tag{$C1\prime\prime$}\\
-\frac{2}{3}<\operatorname{Re}(s_{1i})<-\frac{1}{2}\text{,}\tag{$%
C2\prime\prime$}\\
-\frac{2}{3}<\operatorname{Re}(s_{\left(  N-1\right)  i})<-\frac{1}{2}%
.\tag{$C3\prime\prime$}%
\end{gather}
We notice that the solution set of conditions $C1\prime\prime$--$C3\prime
\prime$ is a non-empty open and connected subset in $\mathbb{R}^{D}$. We now
verify that the conditions $C1\prime\prime$--$C2\prime\prime$ imply conditions
$C1^{\prime}$-$C4^{\prime}$. First, consider $J\subseteq T$ such that $|J|=1.$
We can assume that $J=\left\{  i_{0}\right\}  $ for some $i_{0}\in T$. By
conditions $C1\prime\prime$--$C3\prime\prime$, we have%
\begin{equation}
1+\operatorname{Re}\left(  s_{1i_{0}}\right)  +\operatorname{Re}s_{\left(
N-1\right)  i_{0}}<1-1/2-1/2=0,\label{Eq_1}%
\end{equation}%
\begin{equation}
\sum_{2\leqslant i_{0}<j\leq N-2}\operatorname{Re}\left(  s_{i_{0}j}\right)
+\sum_{\substack{2\leq i<i_{0}\leq N-2,\text{ }\\i\in T\backslash
J}}\operatorname{Re}\left(  s_{ii_{0}}\right)  <0,\label{Eq_2}%
\end{equation}
thus, \ $C1^{\prime}$ follows from (\ref{Eq_1}) and (\ref{Eq_2}). Conditions
$C2^{\prime}$, $C3^{\prime}$ and $C4^{\prime}$ follow directly from
$C1^{\prime\prime}$--$C3^{\prime\prime}$.

We now consider $J\subseteq T$ such that $|J|\geq2$. Condition $C1^{\prime}$
is obtained with a similar calculation to (\ref{Eq_1}) and (\ref{Eq_2}). Now,
by condition $C1^{\prime\prime}$, we get%
\[
\left\vert J\right\vert -1+\sum_{2\leq i<j\leq N-2,\text{ }i,j\in
J}\operatorname{Re}\left(  s_{ij}\right)  >\left\vert J\right\vert -1-\frac
{2}{3}>\left\vert J\right\vert -\frac{5}{3}>0,
\]
which implies $C2^{\prime}$. We now verify Condition $C3^{\prime}$. Let
$t\in\left\{  1,N-1\right\}  $, by using conditions $C2\prime\prime$ and
$C3\prime\prime$,
\begin{align*}
&  \left\vert J\right\vert +\sum_{i\in S}\operatorname{Re}\left(
s_{ti}\right)  +\sum_{2\leq i<j\leq N-2,\text{ }i,j\in J}\operatorname{Re}%
\left(  s_{ij}\right)  \\
&  >\left\vert J\right\vert -\frac{2}{3}|S|-\frac{2\left\vert \left(
i,j\right)  ;2\leq i<j\leq N-2,\text{ }i,j\in J\right\vert }{3N_{1}}\\
&  \geq\left\vert J\right\vert -\frac{2}{3}|S|-\frac{2}{3}\text{.}%
\end{align*}
There are two cases. First, $\left\vert S\right\vert =1$. In this case
$\left\vert J\right\vert -\frac{2}{3}|S|-\frac{2}{3}>0$. If $\left\vert
S\right\vert \geq2$, by using $\frac{-2}{3}\left\vert S\right\vert -\frac
{2}{3}\geq-\left\vert S\right\vert $ and $\left\vert J\right\vert
\geq\left\vert S\right\vert ,$ then $\left\vert J\right\vert -\frac{2}%
{3}|S|-\frac{2}{3}\geq\left\vert J\right\vert -\left\vert S\right\vert \geq0$.

Finally, conditions $C4^{\prime}$ follows from conditions $C1\prime\prime
$-$C3\prime\prime$. Therefore, $H_{0}(\mathbb{R})$ is convex and bounded set
with non-empty interior, and $H_{0}(\mathbb{C})$ contains an open and
connected subset of $\mathbb{C}^{D}$. Finally, since conditions $C1^{\prime}%
$-$C4^{\prime}$ imply conditions $C1$-$C6$, we conclude that $H_{0}%
(\mathbb{R})\subset H(\mathbb{R})$ and that $H_{0}(\mathbb{C})\subset
H(\mathbb{C})$.

In the case $N=4$, $|T|=1$, the verification of the claim is straightforward.
\end{proof}

\begin{theorem}
\label{TheoremA}(1) The $p$\textit{-adic open string }$N$\textit{-point zeta
function}, $Z^{\left(  N\right)  }\left(  \boldsymbol{s}\right)  $, gives rise
to a holomorphic function on $H(\mathbb{C})$, which contains an open and
connected subset of $\mathbb{C}^{D}$. Furthermore, $Z^{\left(  N\right)
}\left(  \boldsymbol{s}\right)  $ admits an analytic continuation to
$\mathbb{C}^{D}$, denoted also as $Z^{\left(  N\right)  }\left(
\boldsymbol{s}\right)  $, as a rational function in the variables $p^{-s_{ij}%
},i,j\in\left\{  1,\ldots,N-1\right\}  $. The real parts of the poles of
$Z^{\left(  N\right)  }\left(  \boldsymbol{s}\right)  $\ belong to a finite
union of hyperplanes, the equations of these hyperplanes have the form
$C1$-$C6$ \ with the symbols `$< $', `$>$' replaced by `$=$'. (2) If
$\boldsymbol{s}=\left(  s_{ij}\right)  \in\mathbb{C}^{D}$, with
$\operatorname{Re}(s_{ij})\geq0$ for $i,j\in\left\{  1,\ldots,N-1\right\}  $,
then the integral $Z^{\left(  N\right)  }\left(  \boldsymbol{s}\right)  $
diverges to $+\infty$.
\end{theorem}

\begin{proof}
(1) We recall that
\begin{equation}
Z^{(N)}\left(  \boldsymbol{s}\right)  =\sum_{I\subseteq T}Z^{(N)}\left(
\boldsymbol{s};I\right)  =\sum_{I\subseteq T}p^{M(\boldsymbol{s})}Z_{0}%
^{(N)}\left(  \boldsymbol{s};I\right)  Z_{1}^{(N)}\left(  \boldsymbol{s}%
;T\smallsetminus I\right)  , \label{Formula_zeta}%
\end{equation}
see Remark \ref{Nota_Lemma_A}. Now, by Propositions \ref{Proposition_1}%
-\ref{Proposition_2} and Lemma \ref{Lemma_H_0}, for any $I\subseteq T$,
$Z_{0}^{(N)}\left(  \boldsymbol{s};I\right)  $ and $Z_{1}^{(N)}\left(
\boldsymbol{s};T\smallsetminus I\right)  $ are holomorphic functions of
$\boldsymbol{s}\in H_{0}(\mathbb{C})$, which is an open and connected subset,
and consequently the analytic continuations of \ the integrals $Z_{0}%
^{(N)}\left(  \boldsymbol{s};I\right)  $ and $Z_{1}^{(N)}\left(
\boldsymbol{s};T\smallsetminus I\right)  $ and formula (\ref{Formula_zeta})
give rise to an analytic continuation of $Z^{(N)}\left(  \boldsymbol{s}%
\right)  $\ with the announced properties.

(2) It follows from formula (\ref{Formula_zeta}) by Corollary
\ref{Corollary_Lemma_3A} and Remark \ref{Remark_B_Lemma_5}.
\end{proof}

\begin{acknowledgement}
The authors wish to thank to the referees for their careful reviewing of the
original manuscript, and for many useful suggestions and comments that help us
to improve the original manuscript.
\end{acknowledgement}

\end{document}